\documentclass[11pt]{article}
\usepackage[margin=2cm]{geometry}
\usepackage{hyperref}

\usepackage[dvipsnames]{xcolor}

\usepackage{authblk}

\usepackage{soul}

\usepackage{microtype}
\usepackage{amsmath}
\usepackage{amssymb}
\usepackage{amsthm}
\usepackage{thmtools,thm-restate}
\usepackage{tcolorbox}
\usepackage{soul}
\usepackage[shortlabels]{enumitem}
\usepackage{lipsum}
\usepackage{xargs}
\usepackage{dsfont}
\usepackage{tikz}
\usepackage{graphicx}

\usepackage{algorithm}
\usepackage{algpseudocode}

\newtheorem{mainthm}{Main Result (informal)}
\newtheorem{definition}{Definition}
\newtheorem{observation}{Observation}
\newtheorem{corollary}{Corollary}

\title{The Power of Amortized Recourse for Online Graph Problems} 

\author[1]{Hsiang-Hsuan Liu}
\author[1]{Jonathan Toole-Charignon}
\affil[1]{Department of Information and computing sciences, Utrecht University, The Netherlands}

\date{}

\begin{document}

\maketitle

\newcommand{\NOTE}[1]{{\color{Red}{#1}}}
\newcommand{\hhl}[1]{{\color{OliveGreen}{#1}}}
\newcommand{\jtc}[1]{{\color{Blue}{#1}}}
\newcommand{\runtitle}[1]{\textbf{#1}}
\newcommand{\pf}[1]{}
\newcommand{\hhlcmd}[1]{\todo[tickmarkheight=0.1cm,size=tiny,color=green!30]{{\bf HHL: }#1}}
\newcommandx{\hide}[1]{}

\newcommand{\tas}{\texttt{TaS}}
\newcommand{\ralg}{\texttt{REF}}
\newcommand{\LGreedy}{\texttt{$L$-Greedy}}
\newcommand{\PEA}{\texttt{DH}}
\newcommand{\ar}{\text{AR}}
\newcommand{\tr}{\text{TR}}
\newcommand{\ins}{\mathcal{X}}
\newcommand{\alg}{\texttt{ALG}}
\newcommand{\opt}{\texttt{OPT}}
\newcommand{\VC}{\textsc{Vertex\ Cover}}
\newcommand{\MVC}{\textsc{Minimum\ Vertex\ Cover}}
\newcommand{\IS}{\textsc{Independent\ Set}}
\newcommand{\MCM}{\textsc{Maximum\ Cardinality\ Matching}}
\newcommand{\me}{\text{ME}}
\newcommand{\lo}{\text{LO}}
\newcommand{\vm}{V_M}
\newcommand{\vmi}{V_{M(\ins)}}
\newcommand{\dphi}{\Delta\Phi}
\newcommand{\pea}{\PEA}
\newcommand{\typeone}{\textbf{Type-1}}
\newcommand{\typetwo}{\textbf{Type-2}}

\begin{abstract}
In this work, we study online graph problems with monotone-sum objectives. 
We propose a general two-fold greedy algorithm that references yardstick algorithms to achieve $t$-competitiveness while incurring at most $\frac{w_{\text{max}}\cdot(t+1)}{\min\{1, w_\text{min}\}\cdot(t-1)}$ amortized recourse, where $w_{\text{max}}$ and $w_{\text{min}}$ are the largest value and the smallest positive value that can be assigned to an element in the sum.
We further show that the general algorithm can be improved for three classical graph problems by carefully choosing the referenced algorithm and tuning its detailed behavior.
For $\IS$, we refine the analysis of our general algorithm and show that $t$-competitiveness can be achieved with $\frac{t}{t-1}$ amortized recourse. 
For $\MCM$, we limit our algorithm's greed to show that $t$-competitiveness can be achieved with $\frac{(2-t^*)}{(t^*-1)(3-t^*)}+\frac{t^*-1}{3-t^*}$ amortized recourse, where~$t^*$ is the largest number such that $t^*= 1 +\frac{1}{j} \leq t$ for some integer $j$.
For $\VC$, we show that our algorithm guarantees a competitive ratio strictly smaller than~$2$ for any finite instance in polynomial time while incurring at most $3.33$ amortized recourse. We beat the almost unbreakable $2$-approximation in polynomial time by using the optimal solution as the reference without computing it.
We remark that this online result can be used as an offline approximation result (without violating the unique games conjecture~\cite{DBLP:journals/jcss/KhotR08}) to partially improve upon the constructive algorithm of Monien and Speckenmeyer~\cite{DBLP:journals/acta/MonienS85}.
\end{abstract}

\section{Introduction}\label{sec:Intro}
\hide{
In this paper, we study graph optimization problems with monotone-sum objectives in the recourse model (defined later). 
We investigate the amount of amortized recourse that is sufficient for attaining a desirable competitive ratio for an arbitrary monotone-sum problem.
We further study three of monotone-sum problems: \textsc{IndependentSet}, \textsc{MaximumMatching}, and \textsc{VertexCover}.
We provide a more sophisticated analysis on the trade-off between the competitive ratio and the recourse for these three problems.
}

Graph optimization problems serve as stems for various practical problems. 
A solution for such a problem can be described as an assignment from the elements of the problem (e.g. vertices of a graph) to non-negative real numbers such that the constraints between the elements are satisfied.
In the online setting, the most considered models are the \emph{vertex-arrival} and \emph{edge-arrival} models. That is, the graph is revealed vertex-by-vertex or edge-by-edge, and once an element arrives, the \emph{online algorithm} has to immediately make an irrevocable decision on the new element. The performance of an online algorithm is measured by \emph{competitive ratio} against the optimal offline solution. 
Many graph optimization problems are non-competitive: the larger the input size, the larger the competitive ratio of any deterministic online algorithm. In other words, a non-competitive problem has no constant-competitive online algorithm.

The pure online model is pessimistic, in that altering decisions may be possible (albeit expensive) or limited knowledge about the future may be available in the real world. 
In this work, we investigate online graph optimization problems in the \emph{recourse} model. That is, decisions made by the online algorithm can be revoked. 
In particular, we aim at finding out the amount of amortized recourse that is sufficient and/or necessary for attaining a desirable competitive ratio for a given problem.

\smallskip

\runtitle{Uncertainty and amortized recourse.}
The competitive ratio can be seen as quantification of how far the quality of an online algorithm's solution is from that of a conceptual optimal offline algorithm that has complete knowledge of the input and unlimited computational power.
Therefore, the non-competitiveness of graph optimization problems suggests that uncertainty of the input is critical to these problems. 
However, the online algorithm may perform better when the irrevocability constraint is relaxed or knowledge about future inputs is available.
It is intriguing to investigate to what extent these problems remain non-competitive under these conditions, in particular to determine how much revocability or knowledge the online algorithm needs in order to attain a desirable competitive ratio.

Beyond the practical motivation of relaxing irrevocability of online algorithms' decisions, amortized recourse also provides insight on how a given online problem is affected by uncertainty. 
In particular, it captures how rapidly the structure of the offline optimal solution can change: the fewer elements required to do so, the larger the amortized recourse. 
Furthermore, the impact of uncertainty is directly correlated with this idea: the faster the optimal solution can change, the more impact uncertainty on future inputs will have.
Different problems may attain constant competitive ratios using different amounts of (amortized) recourse, which implies variability in the impact of uncertainty.
For example, to attain a constant competitive ratio, one needs exactly $O(\log n)$ recourse per edge for min-cost bipartite matching~\cite{DBLP:conf/approx/MegowN20}, while one only needs a constant amount of recourse per element for maximum independent set and minimum vertex cover~\cite{DBLP:journals/algorithmica/BoyarFKL22}.

\smallskip

\runtitle{Online monotone-sum problems.} 
We study \emph{online} graph problems in the vertex-arrival or the edge-arrival models. 
Along with the newly-revealed element, which can be a vertex or an edge according to the arrival model, there may be constraints imposed upon some subset of the currently-revealed elements that a feasible solution should satisfy.
An algorithm aims at finding a feasible solution that maximizes (or minimizes) the objective.
A problem is a \emph{sum} problem if the objective is a sum of the values assigned to each element. 
If the value of the optimal solution of an instance is always greater than or equal to that of a subset of the instance, then the problem is a \emph{monotone} problem.\footnote{The \textsc{Dominating Set} and \textsc{Matching with delays} problems are sum problems but not monotone. The \textsc{Coloring Problem} is monotone but not a sum problem.}

An online algorithm makes decisions upon arrival of each element.
In the \emph{recourse} model, the online algorithm can also revoke an earlier decision that it made and pay for the revocation. We aim to reduce the competitive ratio with as little total recourse (i.e. as few revocations) as possible.

\smallskip

\runtitle{Our contribution.}
We propose a general online algorithm $\tas_t$ that provides a trade-off between amortized recourse and competitive ratio for arbitrary monotone-sum graph problems. In particular, we consider two measurements of recourse cost: number of reassigned elements, or the amount of change in the reassigned values. Our result works for both unweighted and weighted problems, and it even works for fractional optimization problems, where the smallest non-zero value assigned to a single element can be a real number between $0$ and $1$. The following is the main result of our work, where the bound of amortized recourse works for both measurements of recourse cost (Theorem~\ref{thm:general_UB} and Corollary~\ref{cor:fractionalUB}).

\begin{mainthm}
Using an optimal algorithm (resp. an incremental $\alpha$-approximation algorithm, defined formally in Section~\ref{sec:general}) as the yardstick, $\tas_t$ is $t$-competitive (resp. $(t\cdot \alpha)$-competitive) and incurs at most $\frac{w_{\text{max}}\cdot(t+1)}{\min\{1, w_\text{min}\}\cdot(t-1)}$ amortized recourse for any monotone-sum graph problem where $w_\text{max}$ and $w_\text{min}$ are the maximum and minimum non-zero values that can be assigned to an element.\footnote{The bound of amortized recourse $\frac{w_{\text{max}}\cdot(t+1)}{w_\text{min}\cdot(t-1)}$ is larger when the elements can be assigned minimum non-zero values smaller than $1$. For example, the fractional $\VC$ problem in~\cite{DBLP:conf/icalp/WangW15}.}
\end{mainthm}

\hide{
Note that the $\tas_t$ algorithm can use any algorithm as a yardstick, as long as it is \emph{incremental}. That is, the algorithm's solution value for an instance is always greater than or equal to the value for its subinstances (e.g., the approximation algorithm proposed in~\cite{DBLP:journals/bit/BoppanaH92}). 
}


$\tas_t$ is two-fold greedy. 
First, it 
assigns the value \emph{greedily} once an element arrives. 
Second, the algorithm aligns its solution to the yardstick solution \emph{completely} and incurs recourse when the current solution fails to be $t$-competitive against the yardstick solution.

In general, the $\tas_t$ algorithm works for any optimization problem.
The challenge is to bound the amortized recourse that it incurs, as the complete alignment may require a vast amount of recourse.
By looking closer at a specific problem, we can show a tighter bound on the amount of recourse needed. We use a sophisticated analysis for the $\IS$ problem and improve the recourse bound (Theorem~\ref{thm:ISUB}).

The two-fold greedy algorithm may perform better when the greediness is relaxed.
Moreover, by choosing different yardstick algorithms and tuning the alignment to the yardstick carefully, the amortized recourse can be further reduced.
We show that for the $\MCM$ problem, partially aligning to the yardstick solution is more recourse-efficient (Theorem~\ref{thm:matchingUB}).

For the $\VC$ problem, we show that a special version of $\tas_t$ with $t = 2-\frac{2}{\opt}$ incurs a very small amount of amortized recourse (Theorem~\ref{thm:VCUB}) and is $(2 - \frac{2}{\opt})$-competitive, where $\opt$ is the size of the optimal vertex cover\footnote{Note that over all instances, $\opt$ can be arbitrarily large. Thus, there is no $\varepsilon > 0$ for which $2 - \frac{2}{\opt} \le 2 - \varepsilon$ over all instances. Therefore, our result does not violate the unique games conjecture~\cite{DBLP:journals/jcss/KhotR08}.} (Theorem~\ref{thm:VCCompRatio_new}).
Our algorithm uses an optimal solution as a yardstick. The key to the polynomial time complexity is that instead of explicitly finding the yardstick assignment, we show that the yardstick cannot be too ``far'' from our solution at any moment if the target competitive ratio $2-\frac{2}{\opt}$ is not already achieved. More specifically, by restricting the range of greedy choice, we can show that the yardstick solution can be aligned partially within a constant amount of amortized recourse.
Thus, our result breaks the almost unbreakable $2$-approximation for the $\VC$ problem and improves upon that of Monien and Speckenmeyer~\cite{DBLP:journals/acta/MonienS85} for a subset of the graphs containing odd cycles of length no less than $2k + 3$ (for which $2 - \frac{2}{\opt} < 2 - \frac{1}{k + 1}$), using an algorithm that is also constructive. 

Our results are summarized in Table~\ref{tb:results}, which illustrates the power of amortization.
\begin{table}[h!]
\begin{tabular}{|l|l|l|} 
\hline
& \begin{tabular}[c]{@{}l@{}}(Competitive ratio, \\worst case recourse)\end{tabular} & \begin{tabular}[c]{@{}l@{}}(Competitive ratio, amortized recourse)\end{tabular} \\ \hline
\hline
monotone-sum problems  && 
\begin{tabular}[c]{@{}l@{}}
$(t\alpha, \frac{w_{\text{max}}\cdot(t+1)}{\min\{w_{\text{min}}, 1\}\cdot(t-1)})$ (Theorem~\ref{thm:general_UB}, Corollary~\ref{cor:fractionalUB},\\P with incremental $\alpha$-approximation algorithms) 
\end{tabular}\\ \hline
Maximum Independent Set &$(2.598, 2)$~\cite{DBLP:journals/algorithmica/BoyarFKL22} & \begin{tabular}[c]{@{}l@{}} $(t,\frac{t}{t-1})$ (Theorem~\ref{thm:ISUB}) \\$(2.598, 1.626)$ (Theorem~\ref{thm:ISUB}) 
\end{tabular} \\ \hline
Maximum Matching & \begin{tabular}[c]{@{}l@{}} $(k, O(\frac{\log k}{k})+1))$~\cite{DBLP:journals/jco/AngelopoulosDJ20}\\$(1.5, 2)$~\cite{DBLP:journals/algorithmica/BoyarFKL22}\end{tabular} & \begin{tabular}[c]{@{}l@{}}
$(t,\frac{(2-t^*)}{(t^*-1)(3-t^*)}+\frac{t^*-1}{3-t^*})$ (Theorem~\ref{thm:matchingUB}, P)
\\($(1.5, 1)$ with $t^* = 1.5$)
\end{tabular} \\ \hline
Minimum Vertex Cover            & $(2, 1)$~\cite{DBLP:journals/algorithmica/BoyarFKL22}  & $(2-\frac{2}{\opt}, \frac{10}{3})$ (Theorem~\ref{thm:VCUB}, P)
\\ \hline
\end{tabular}
\caption{Summary of our results. Note that $t$ can be any real number larger than $1$. For Maximum Matching, $t^*$ is the largest number such that $t^* \leq t$ and $t^* = 1 + \frac{1}{j}$ for some integer $j$. The note P means that the algorithm is a polynomial-time online algorithm.}
\label{tb:results}
\end{table}
\vspace{-5ex}

\smallskip

\textbf{Related work.}
The closest previous result is the work by Boyar et al.~\cite{DBLP:journals/algorithmica/BoyarFKL22}. The authors investigated the $\IS$, $\MCM$, $\VC$, and \textsc{Minimum Spanning Forest} problems, which are all non-competitive in the pure online model. The authors showed that the competitive ratio of these problems can be massively reduced to a constant by incurring at most $2$ recourse for any single element.
Note that the bounds of the worst case recourse are upper bounds of the amortized recourse. Moreover, the algorithms in~\cite{DBLP:journals/algorithmica/BoyarFKL22} incur at least $1.5$ amortized recourse for the $\MCM$ problem and at least $0.5$ amortized recourse for the $\VC$ problem.

There is a line of research on online matching problems with recourse.
Angelopoulos et al.~\cite{DBLP:journals/jco/AngelopoulosDJ20} studied a more general setting for $\MCM$ and showed that given that no element incurs more than $k$ recourse, there exists an algorithm that attains a competitive ratio of $1 + O(1/\sqrt{k})$.
Megow and N\"{o}lke~\cite{DBLP:conf/approx/MegowN20} showed that for the \textsc{Min-Cost Bipartite Matching} problem, constant competitiveness is achievable with amortized recourse $O(\log n)$, where $n$ is the number of requests.
Bernstein et al.~\cite{DBLP:journals/jacm/BernsteinHR19} showed that there exists an algorithm that achieves $1$-competitiveness with $O(\log^2 n)$ amortized recourse for the \textsc{Bipartite Matching} problem, where $n$ is the number of vertices inserted. The result also shows that to achieve $1$-competitiveness for $\VC$, any online algorithm needs at least $O(n)$ amortized recourse per vertex. 

In addition, there has been extensive work on online algorithms in the recourse model
for a variety of different problems. 
For amortized recourse, studied problems include online bipartite matching~\cite{DBLP:journals/jacm/BernsteinHR19}, graph coloring~\cite{DBLP:conf/swat/BosekDFPZ20}, minimum spanning tree and traveling salesperson~\cite{DBLP:journals/siamcomp/MegowSVW16}, Steiner tree~\cite{DBLP:journals/siamcomp/GuG016}, online facility location~\cite{DBLP:conf/esa/CyganCMS18}, bin packing~\cite{DBLP:journals/corr/abs-1711-02078}, submodular covering~\cite{DBLP:conf/focs/GuptaL20}, and constrained optimization~\cite{DBLP:journals/ipl/AvitabileMP13}.



Graph problems model various real-world issues whose performance guarantees are often abysmal, as they are notoriously non-competitive in the pure online model.
Prior work has shown curiosity about the conditions under which these problems become competitive, and these problems have been investigated under different models out of both practical and theoretical interests. 
Other than the recourse model, considered models include paying for a delay in the timing of decision making to achieve a better solution~\cite{ DBLP:conf/focs/AzarT20, DBLP:conf/waoa/BienkowskiKLS18}. 
Another model for delayed decision making is the reordering buffer model~\cite{DBLP:conf/esa/AlbersS17}, where the online algorithm can delay up to $k$ decisions by storing the elements in a size-$k$ buffer. 

The impact of extra knowledge about the input has also been studied. For example, once a vertex arrives, the neighborhood is known to the algorithm~\cite{DBLP:conf/mfcs/HarutyunyanPR21}.
In the lookahead model, an online algorithm is capable of foreseeing the next events~\cite{DBLP:conf/esa/AlbersS17}. 
Predictions provided by machine learning are also considered for graph problems~\cite{DBLP:conf/soda/AzarPT22}.
Finally, there are also works where the integral assignment restrictions are relaxed for vertex cover and matching problems~\cite{DBLP:conf/icalp/WangW15}.




Another major area of related work for practically any problem considered in the online model is polynomial-time approximation algorithms for the equivalent problem in the offline setting.
The link between the two is particularly salient when considering a polynomial-time online algorithm, as this online algorithm can also be run in polynomial time in the offline setting by processing the graph as if it were revealed in an online manner.

In the case of minimum vertex cover, assuming the unique games conjecture, it is not possible to obtain an approximation factor of $(2 - \varepsilon)$ for fixed $\varepsilon > 0$~\cite{DBLP:journals/jcss/KhotR08}.
However, results have been obtained for parameterized $\varepsilon$.
In particular, Halperin~\cite{DBLP:journals/siamcomp/Halperin02} showed an approximation factor of $2-(1-o(1))\frac{2 \ln \ln \Delta}{\ln \Delta}$ on graphs with maximum degree $\Delta$, and Karakostas~\cite{DBLP:journals/talg/Karakostas09} showed an approximation factor of $2 - \theta(\frac{1}{\sqrt{{\rm log} n}})$.
Both of these results use semidefinite relaxations of the problem, whereas Monien and Speckenmeyer~\cite{DBLP:journals/acta/MonienS85} had previously used a constructive approach to show an approximation factor of $2 - \frac{1}{k + 1}$ for graphs without odd cycles of length at most $2k + 1$.

\textbf{Paper organization.}
Section~\ref{sec:general} defines monotone-sum graph problems and the amortized recourse model. We propose a general algorithm $\tas_t$ for finding the trade-off between the desired competitive ratio and the amortized recourse needed. 
Section~\ref{sec:IS} provides a refined analysis on the $\tas_t$ algorithm on the $\IS$ problem.
Section~\ref{sec:Matching} discusses an existing algorithm~\cite{DBLP:journals/jco/AngelopoulosDJ20}, which is a variant of $\tas_t$ algorithm, for the $\MCM$ problem that is less greedy in aligning its solution and obtains a better trade-off.
Section~\ref{sec:VC} introduces a polynomial-time version of $\tas_t$ algorithm for the $\VC$ problem that limits both greedy aspects.
This algorithm can also be used as a novel offline approximation algorithm for certain graph classes.
Due to space constraints, proofs for all lemmas and theorems can be found in the appendix. We also provide proof ideas for some theorems in the main text.


\section{Monotone-Sum Graph Problems and a General Algorithm}\label{sec:general}
For an \emph{online graph problem} $Q$ on a graph $G= (V,E)$, which is unknown a priori, we consider either the \emph{vertex-arrival} model or the \emph{edge-arrival} model. 
In the vertex-arrival model (resp. edge-arrival model), the elements in $V$ (resp. elements in $E$) arrive one and a time, and an algorithm has to assign each element a non-negative value in $[0, w_\text{max}]$ such that the assignment satisfies certain properties associated with $Q$.
Formally, the \emph{assignment} is defined as $\mathcal{A}: \ins \rightarrow \mathbb{R}^+$, where $\ins$ is $V$ or $E$, such that $\mathcal{A}(\ins)$ satisfies a set of properties~$\mathcal{P}_Q$. 
The \emph{value} of a feasible assignment $\mathcal{A}$ is defined as a function $value: \ins\times\mathcal{A}(\ins) \to \mathbb{R}^+$, which should be minimized or maximized as appropriate. 
In this work, we focus on the problems with \emph{sum} objectives, that is, $value(\ins, \mathcal{A}(\ins)) = \sum_{x\in \ins} \mathcal{A}(x)$. 
Moreover, we concern ourselves about the impact of lacking information on the optimality of the solution. Therefore, we consider monotone sum graph problems where given a feasible assignment and a newly-arrived element, there is always a value in $[0, w_\text{max}]$ that can be assigned to the new element such that the new assignment is feasible.\footnote{Classical graph problems such as $\IS$, $\MCM$, and $\VC$ all satisfy this property.}

We denote the assignment on input $\ins$ returned by the algorithm $\alg$ by $\alg(\ins)$. 
We abuse the notation $\ins$ to denote the graph revealed by the input $\ins$.
We further abuse notation and denote the total value of the assignment by $\alg(\ins)$ as well. That is, $\alg(\ins) = \sum_{x_i \in \ins} \alg(x_i)$. When the context is clear, the parameter $\ins$ is dropped.

We study the family of monotone-sum graph problems, which is defined as follows. Similarly, we define the family of incremental algorithms.
Note that a monotone-sum problem can be a maximization or a minimization problem.

\begin{definition}
The \emph{projection} of an assignment $\mathcal{A}(G)$ on an induced subgraph $H$ of $G$ assigns to each element in $H$ the same value that $\mathcal{A}(G)$ does in $G$.
\end{definition}

\begin{definition}
\emph{\textbf{Monotone-sum graph problems.}} A sum problem is \emph{monotone} if for any graph $G$ and any induced subgraph $H$ of $G$, 1) the projection of any feasible assignment $\mathcal{A}(G)$ on $H$ is also feasible, and 2) $\opt(H) \leq \opt(G)$, where $\opt$ is an optimal solution.
\end{definition}

\begin{definition}
\emph{\textbf{Incremental algorithms.}} An algorithm $\alg$ is \emph{incremental} if for any graph~$G$ corresponding to the instance $\ins$ and any induced subgraph $H$ of $G$, $\alg(H) \leq \alg(G)$. Furthermore, the projection of $\alg(\ins)$ on a prefix $\ins^\prime$ of instance $\ins$ does not have a better objective value than the assignment $\alg(\ins^\prime)$.\footnote{For example, the Ramsey algorithm in~\cite{DBLP:journals/bit/BoppanaH92} is an incremental algorithm. Also note that any online algorithm is an incremental algorithm.}
\end{definition}


In this work, the performance of an online algorithm is measured by the \emph{competitive ratio}. An online algorithm $\alg$ attains a competitive ratio of $t$ if $\max\{\frac{\alg(\ins)}{\opt(\ins)}, \frac{\opt(\ins)}{\alg(\ins)}\}\leq t$ for any instance $\ins$, where $\opt$ is the optimal offline algorithm that knows all information necessary for solving the problem.
In the recourse model, the online algorithm can revoke its decisions and incurs \emph{recourse cost}. There are two types of recourse cost considered in this paper:
\begin{itemize}
    \item $\typeone$: The recourse cost is defined as the \emph{number} of elements which assignment values are changed. Formally, $\sum_{x_i \in \ins} \mathds{1}[A_1(x_i)\neq A_2(x_i)]$ when an assignment on instance $\ins$ is changed from $A_1(\ins)$ to $A_2(\ins)$.
    \item $\typetwo$: The recourse cost is defined as the \emph{amount of change} of the assignment value. Formally, $\sum_{x_i \in \ins} |A_1(x_i)-A_2(x_i)|$ when an assignment on instance $\ins$ is changed from $A_1(\ins)$ to $A_2(\ins)$.
\end{itemize}
We study the trade-off between the competitive ratio and the \emph{amortized recourse}. That is, the total incurred recourse cost divided by the number of elements that should be assigned a value in the final instance. 
We define a family of algorithms for monotone-sum problems.


\hide{
\hhl{The monotone-sum problem property captures many classical graph optimization problems such as \textsc{Independent} \textsc{Set}, $\MCM$, and $\VC$. The three problems can be interpreted as a special case of general monotone-sum problems as follows.}

\runtitle{$\IS$ problem in vertex-arrival model.} Vertices arrives one at a time and should be assigned a value $0$ or $1$. Once a vertex is revealed, the edges between it and its previously-revealed neighbors are known. The goal is to find a maximum value assignment such that for any edge, the sum of values assigned to the two endpoints is at most $1$.

\runtitle{$\MCM$ problem in vertex/edge-arrival model.} Edges or vertices arrives one at a time and each of the edges should be assigned a value $0$ or $1$. The goal is to find a maximum value assignment such that for any vertex, the sum of values assigned to its incident edges is at most $1$.

\runtitle{$\VC$ problem in vertex-arrival model.} Vertices arrives one at a time and should be assigned a value $0$ or $1$. Once a vertex is revealed, the edges between it and its previously-revealed neighbors are known. The goal is to find a minimum value assignment such that for any edge, the sum of values assigned to its two endpoints is at least $1$.

\smallskip

Since the available value for each element is either $0$ or $1$ in these three problems, we say that an element is \emph{accepted} if it is assigned a value $1$. Similarly, an element is \emph{rejected} if it is assigned a value $0$. An element is \emph{late-accepted} if its value is changed from $0$ to $1$ after its arrival, and \emph{late-rejected} if its value is changed from $1$ to $0$ after its arrival.
\hhl{Furthermore, since the value for any element only changes between $0$ and $1$, the $\typeone$ recourse cost and $\typetwo$ recourse cost are equivalent in these three problems.}
}

\runtitle{Target-and-Switch ($\tas_t$) algorithm.}
The $\tas_t$ algorithm uses a yardstick algorithm $\ralg$ as a reference, where the yardstick can be the optimal algorithm or an incremental $\alpha$-approximation algorithm.
Throughout the process, $\tas_t$ keeps track of the yardstick solution value.
Once a new element arrives, $\tas_t$ greedily assigns a feasible value\footnote{Note that there always exists a value such that the new assignment is feasible since the problem is monotone.} to the newly-revealed element if this assignment remains $t$-competitive relative to the yardstick algorithm's solution. 
Otherwise, $\tas_t$ \emph{switches} its assignment to the one by the yardstick algorithm and incurs recourse. (See Algorithm~\ref{alg:tas} in Appendix.)

Now, we show that the $\tas_t$ algorithm achieves the desired competitive ratio $t$ with at most polynomial of $t$ amortized recourse. In our analysis, we use the following observation heavily (including for Theorem~\ref{thm:general_UB}).

\begin{observation}\label{obs:sum-max}
For all $x_i \geq 0$ and $y_i > 0$, $\frac{\Sigma_i x_i}{\Sigma_i y_i} \leq \max_i \frac{x_i}{y_i}$.
\end{observation} 

\begin{restatable}{theorem}{theoremGeneralUB}
\label{thm:general_UB}
Using an optimal algorithm (resp. incremental $\alpha$-approximation algorithm) as the yardstick, $\tas_t$ is $t$-competitive (resp. $(t\cdot \alpha)$-competitive) and incurs at most $\frac{w_{\text{max}}\cdot(t+1)}{t-1}$ $\typetwo$ amortized recourse for any monotone-sum graph problem where $w_\text{max}$ is the maximum value that can be assigned to an element. The bound also works for $\typeone$ amortized recourse.
\end{restatable}
\begin{proof} \textbf{(Ideas.)}
We can show that any optimal solution satisfies the incremental property (see the full version) and thus can be seen as an incremental $1$-approximation algorithm.

Since recourse is incurred only at the moments when a switch happens in the $\tas_t$ algorithm, we partition the process of the algorithm into \emph{phases} according to the switches. Phase $i$ consists all the events after the $(i-1)$-th switch until the $i$-th switch. 
By Observation~\ref{obs:sum-max}, the amortized recourse for the whole instance is bounded by the maximum amortized recourse incurred within a phase.
Therefore, we consider the amortized recourse incurred by the $(i+1)$-th switch for arbitrary $i\geq 0$.

Let $\ralg_i$ and $\tas_i$ denote the value of the yardstick algorithm's solution and the $\tas_t$ algorithm's solution right \emph{after} the $i$-th switch, respectively. By the $\tas_t$ algorithm, $\tas_i = \ralg_i$. 
Let $\alg$ denote the value of $\tas_t$'s solution right \emph{before} the arrival of $x$, which triggers the $(i+1)$-th switch. 
The total $\typetwo$ recourse cost is at most $\alg + \ralg_{i+1}$ (where the $\tas_t$ algorithm changes the value on every element to zero and then changes it to the $\ralg$ assignment).

The main ingredients for the proof are:
\begin{itemize}
    \item \textbf{Property $1$}: Monotonicity of the problem and the incremental nature of $\ralg$ implies that $\ralg_i \leq \ralg_{i+1}$.
    \item \textbf{Property $2$}: The incremental nature of $\tas_t$ during a phase implies that $\alg \geq \ralg_i$.
    \item \textbf{Property $3$}: By the switching condition of $\tas_t$, $\alg <\ralg_{i+1}/t$ for maximization problems, and $\alg+w_\text{max} > t\cdot\ralg_{i+1}$ for minimization problems.
\end{itemize}

\smallskip

\runtitle{Maximization problems.} 
By \textbf{Property $1$} and the fact that the assigned values are at most $w_\text{max}$, we can show that the adversary needs to release at least $\frac{\ralg_{i+1} - \ralg_i}{w_\text{max}}$ elements such that 
the yardstick assignment value 
increases enough to trigger the switch. 
By \textbf{Property $2$} and \textbf{Property $3$}, $\frac{\ralg_{i+1} - \ralg_i}{w_\text{max}} \geq (1-1/t)\cdot \ralg_{i+1}$.
By \textbf{Property $3$}, the total recourse incurred by the $(i+1)$-th switch is at most $\alg+\ralg_{i+1} < (1+1/t) \cdot \ralg_{i+1}$. 
Hence, the $\typetwo$ amortized recourse incurred in phase $i+1$ is bounded by $\frac{w_\text{max}\cdot(1+1/t) \cdot \ralg_{i+1}}{(1-1/t)\cdot \ralg_{i+1}} = \frac{w_\text{max}\cdot(t+1)}{t-1}$.

\smallskip

\runtitle{Minimization problems.} 
In minimization problems, the $(i+1)$-th switch may be triggered by shifting the $\ralg$ assignment completely but without changing its value. In this case, a massive amount of recourse is incurred by a single input. 
However, we can show by \textbf{Property $3$} that in this case, the $\alg$ value must be large enough to trigger the switch. 
Thus, we can bound the number of elements released during phase $i+1$ by the change of $\tas_t$ assignment's total value. That is, it is at least $\frac{\alg - \tas_i}{w_\text{max}}+1$, where the $1$ is the element which triggers the switching.
By \textbf{Property $1$} and \textbf{Property $2$}, the number is at least $\frac{(1-1/t)\cdot(\alg +w_\text{max})}{w_\text{max}}$.
By \textbf{Property $3$}, the total recourse incurred by the $(i+1)$-th switch is at most $\alg+\ralg_{i+1} < (1+1/t) \cdot \alg+w_\text{max}/t$. 
Therefore, the $\typetwo$ amortized recourse incurred in phase $i+1$ is bounded by $\frac{w_\text{max}\cdot((1+1/t) \cdot \alg+w_\text{max}/t)}{(1-1/t)\cdot(\alg +w_\text{max})} = \frac{w_\text{max}\cdot(t+1)}{t-1}$.
\end{proof}


The yardstick algorithm can be the optimal offline algorithm. 
Since the problem is monotone, our algorithm can be $t$-competitive for arbitrary $t >1$.
Furthermore, if we apply a polynomial-time incremental $\alpha$-approximation algorithm as the yardstick, then our algorithm also runs in polynomial time.

\hide{
Note that the amount of recourse in the analysis is measured by the \emph{total amount of change} among all elements. That is, for a single element $x$, its incurred recourse is $|\mathcal{A^\prime}(x) - \mathcal{A}(x)|$, where $\mathcal{A}^\prime$ is the adapted assignment.
This amortized recourse is always an upper bound of the one if we measure the recourse by the \emph{number} of changed element values when the minimum non-zero value assigned to the elements is at least $1$. 
}

The results work for weighted versions of problems, and it also work for fractional assignment problems, where the value assigned to any element is in $[0,1]$ (for example, the fractional $\VC$ problem in~\cite{DBLP:conf/icalp/WangW15}). In this case, the $\typetwo$ amortized recourse is bounded above by the $\typeone$ amortized recourse:

\begin{restatable}{corollary}{corFractionalUB}
\label{cor:fractionalUB}
For a fractional monotone-sum problem, $\tas_t$ is $(t\cdot \alpha)$-competitive and incurs at most $\frac{t+1}{w_{\text{min}}\cdot(t-1)}$ $\typeone$ amortized recourse using an incremental $\alpha$-approximation algorithm as the yardstick. The bound also works for $\typetwo$ amortized recourse.
\end{restatable}


The monotone-sum problem property captures many classical graph optimization problems such as \textsc{Independent} \textsc{Set}, $\MCM$, and $\VC$. The three problems can be interpreted as a special case of general monotone-sum problems as follows.

\runtitle{$\IS$ problem in vertex-arrival model.} Vertices arrive one at a time and should be assigned a value $0$ or $1$. Once a vertex is revealed, the edges between it and its previously-revealed neighbors are known. The goal is to find a maximum value assignment such that for any edge, the sum of values assigned to the two endpoints is at most $1$.

\runtitle{$\MCM$ problem in vertex/edge-arrival model.} Edges or vertices arrive one at a time and each of the edges should be assigned a value $0$ or $1$. The goal is to find a maximum value assignment such that for any vertex, the sum of values assigned to its incident edges is at most $1$.

\runtitle{$\VC$ problem in vertex-arrival model.} Vertices arrive one at a time and should be assigned a value $0$ or $1$. Once a vertex is revealed, the edges between it and its previously-revealed neighbors are known. The goal is to find a minimum value assignment such that for any edge, the sum of values assigned to its two endpoints is at least $1$.

\smallskip

Since the available value for each element is either $0$ or $1$ in these three problems, we say that an element is \emph{accepted} if it is assigned a value $1$. Similarly, an element is \emph{rejected} if it is assigned a value $0$. An element is \emph{late-accepted} if its value is changed from $0$ to $1$ after its arrival, and \emph{late-rejected} if its value is changed from $1$ to $0$ after its arrival.
Furthermore, since the value for any element only changes between $0$ and $1$, the $\typeone$ recourse cost and $\typetwo$ recourse cost are equivalent in these three problems. Therefore, we have the following corollary.
\begin{corollary}\label{col:threeproblems}
For $\IS$, $\MCM$, and \textsc{Vertex}\textsc{Cover} problems, the $\tas_t$ algorithm attains competitive ratio $t>1$ while incurring at most $\frac{t+1}{t-1}$ ($\typeone$ or $\typetwo$) amortized recourse.
\end{corollary}

\section{Maximum Independent Set}\label{sec:IS}
For the maximum independent set problem in the vertex-arrival model, the algorithm proposed by Boyar et al. incurs at most $2$ amortized recourse while maintaining a competitive ratio of $2.598$~\cite{DBLP:journals/algorithmica/BoyarFKL22}.
By Theorem~\ref{thm:general_UB}, the general $\tas_t$ algorithm incurs at most $\frac{t+1}{t-1}$ amortized recourse and guarantees a competitive ratio of $t$.
In this section, we show that the amortized recourse incurred by $\tas_t$ is even smaller by a more sophisticated analysis. 

\begin{restatable}{lemma}{lmISreduc}
\label{lm:reduction}
\runtitle{\emph{(Instance reduction)}}
For any instance $(G, \sigma)$ of the maximum independent set problem, there exists an instance $(G', \sigma')$ for which any newly revealed vertex is either accepted by $\tas_t$ or is part of the optimal offline solution when $\tas_t$ incurs its next switch, but not both, such that the amortized recourse for $(G', \sigma')$ is at least that for $(G, \sigma)$.
\end{restatable}

Using Lemma~\ref{lm:reduction}, we can bound above the amortized recourse incurred by $\tas_t$ against any reduced instance, and thus the amortized recourse incurred against any instance.

\begin{restatable}{theorem}{thmISub}
\label{thm:ISUB}
For the maximum independent set problem, given a target competitive ratio~$t > 1$, $\tas_t$ is $t$-competitive while incurring at most $\frac{t}{t-1}$ amortized recourse. 
\end{restatable}

\begin{proof} \textbf{(Ideas.)}
We show that, for any reduced instance from Lemma~\ref{lm:reduction}, $\tas_t$ will incur at most $\frac{t}{t-1}$ amortized recourse, and thus that this upper bound holds for any instance. To do this, we use the same phase partition argument as in the proof of Theorem~\ref{thm:general_UB} combined with Observation~\ref{obs:sum-max}.
We consider a scheme in which each newly-revealed vertex carries budget~$B$, and the vertices revealed in phase $i+1$ must pay the full cost of the recourse incurred by switch $i+1$. 
If the total budget carried by these newly-revealed vertices is at least $ALG + \opt_{i+1}$, the amortized recourse is $B$.

We can show that the number of vertices revealed in phase $i+1$ that are part of $\opt_{i+1}$ is bounded above by $\opt_{i+1} - \opt_{i-1}$, which implies that it is sufficient for the budget to satisfy $B \ge \frac{ALG + \opt_{i+1}}{\opt_{i+1} - \opt_{i-1}}$.
Furthermore, we incorporate both the number of vertices in phase $i+1$ that are accepted by $\tas_t$ and the number of vertices revealed in phase $i$ that are accepted by $\tas_t$ into our analysis and show that the lower bound on the required budget is largest when there are no such vertices.

We conclude that it is sufficient for each newly-revealed vertex to carry budget $B = \frac{t}{t - 1}$. Thus, $\tas_t$ is $t$-competitive while incurring at most $\frac{t}{t - 1}$ amortized recourse.
\end{proof}

\begin{restatable}{theorem}{thmISlb}
For any $1 < t \le 2$, $\varepsilon > 0$, and $t$-competitive deterministic online algorithm, there exists an instance for which the algorithm incurs at least $\frac{1}{t-1} - \varepsilon$ amortized recourse.
\end{restatable}
\hide{
\begin{proof} \textbf{(Ideas.)}
Consider any $t$-competitive online algorithm against an adversary that constructs a complete bipartite graph and only reveals new vertices in the partition which does not contain the algorithm's current solution.
This means that the maximum number of vertices that the algorithm's solution can contain will only increase if the algorithm moves its solution from one partition to the other.
Furthermore, in doing so, the algorithm is forced to late-reject all vertices in its old solution in order to late-accept all vertices in its new solution.

We use the structure of this adversary to show that each partition-changing switch will incur at least $\frac{1}{t}$ recourse amortized over the size of the revealed graph when the switch occurs.
Furthermore, we show that there are at most $t$ times more vertices at switch $i + 1$ than at switch $i$.
From this, we derive a recurrence relation that bounds below the recourse incurred by all switches up to switch $i$ amortized over the size of the revealed graph when switch $i$ occurs: $f(i) = \frac{1}{t} + \frac{1}{t} f(i-1)$.
We further show that we can set the initial value of the recurrence relation as $f(1) = 1$.

Solving this recurrence relation, we obtain that the recourse incurred by all switches up to switch $i$ amortized over the size of the revealed graph when switch $i$ occurs is bounded below by $\frac{(t - 2)(\frac{1}{t})^{i-1} + 1}{t - 1}$. This lower bound applies for any online algorithm against the described adversary when said adversary terminates its input sequence after the algorithm's $i$-th switch.

Therefore, for any $1 < t \le 2$, $\varepsilon > 0$, and $t$-competitive deterministic online algorithm, there exists an instance for which the algorithm incurs at least $\frac{1}{t-1} - \varepsilon$ amortized recourse.
\end{proof}
}

\section{Maximum Cardinality Matching}\label{sec:Matching}

The $\tas_t$ algorithm greedily aligns with the yardstick solution completely and incurs a lot of recourse. 
However, for some of the elements whose value is changed, the alignment may not contribute to the improvement of the competitive ratio as much as the alignment of other elements.
This observation suggests that it may be possible to reduce the amount of amortized recourse while maintaining $t$-competitiveness by switching the solution only \emph{partially} into the yardstick. 
In this section, we show that the $\LGreedy$ algorithm by Angelopoulos et al.~\cite{DBLP:journals/jco/AngelopoulosDJ20}, which is in fact a $\tas_t$ algorithm that uses an optimal solution as the yardstick without aligning to it fully, incurs less amortized recourse for the $\MCM$ problem.

\smallskip

\runtitle{$\LGreedy$ algorithm~\cite{DBLP:journals/jco/AngelopoulosDJ20}.} The algorithm is associated with a parameter $L$. Throughout the process, the $\LGreedy$ algorithm partially switches its solution to the optimal once by eliminating all augmenting paths with length at most $2L+1$. 
That is, it late rejects all the edges selected by itself and late accepts all the edges in the optimal solution on the path. 

After applying late operations on all augmenting paths with at most $2L+1$ edges, every remaining augmenting path has length at least $2L+3$, and the ratio of the $\opt$ solution value to the $\LGreedy$ solution value is $\frac{\opt(P)}{\LGreedy(P)} \leq \frac{L+2}{L+1}$ on the component $P$.
Since the $\MCM$ problem can be solved in $O(n^{2.5})$ time, the following theorem holds by selecting $L = \lceil\frac{1}{t-1}\rceil-1$.
\begin{restatable}{theorem}{thmMCMcr}
\label{thm:MCMcr}
The $\LGreedy$ algorithm returns a valid matching with competitive ratio $\frac{L+2}{L+1}$ in $O(n^{3.5})$ time, where $n$ is the number of vertices in the final graph.
\end{restatable}

\hide{
\begin{proof}
After applying the late operations on all the augmenting path with at most $2L+1$ edges, every remaining augmenting path has length at least $2L+3 = (L+2) + (L+1)$, and the ratio of the $\opt$ size to the $\LGreedy$ size $\frac{\opt(P)}{\LGreedy(P)} \leq \frac{L+2}{L+1}$ on the component $P$.

By selecting $L = \lceil \frac{1}{t-1}\rceil-1$, the $\LGreedy$ algorithm eliminates all augmenting paths that has length at most $2L+1 = 2\lceil \frac{1}{t-1}\rceil-1$. 

Afterwards, every the remaining augmenting path has length at least $2\lceil \frac{1}{t-1}\rceil+1$, and the algorithm attains a competitive ratio at most $\frac{\lceil\frac{1}{t-1}\rceil+1}{\lceil \frac{1}{t-1}\rceil}$.

\runtitle{$t$-competitiveness.} We first show that after switching every augmenting path with length at most $2L+1$, the $\LGreedy$ algorithm attains a competitive ratio of $t$. 
Consider the disjunctive union of the algorithm's matching $M$ and the optimal matching $M^*$, $M \triangle M^*$. The disjunctive union consists of connected components. Let $M_i$ and $M_i^*$ be the edges in the $i$-th component from $M$ and from $M^*$, respectively. According to our algorithm, there are two properties of the connected components:
\begin{enumerate}[label={(\bfseries P\arabic*)}]
\setlength\itemsep{0em}
    \item In any component with odd size, $\frac{|M_i^*|}{|M_i|} \leq \frac{\lceil \frac{1}{t-1}\rceil+1}{\lceil \frac{1}{t-1}\rceil}$.
    \item In any component with even size, $|M_i^*| = |M_i|$.
\end{enumerate}
Therefore, the competitive ratio by the $\LGreedy$ algorithm is at most $\frac{\sum_i |M_i^*|}{\sum_i |M_i|} \leq \max_i \frac{|M_i^*|}{|M_i|} \leq \frac{\lceil\frac{1}{t-1}\rceil+1}{\lceil\frac{1}{t-1}\rceil} = 1 + \frac{1}{\lceil\frac{1}{t-1}\rceil} \leq 1 + \frac{1}{1/(t-1)} = 1+(t-1) = t$.
\end{proof}
}

Since it was shown that to achieve $1.5$-competitiveness, every vertex incurs at most $2$ recourse, we consider a target competitive ratio $1<t<2$ and have the following theorem. Note that $1<t^*<2$, thus $0<\frac{t^*-1}{3-t^*}<1$.
\begin{restatable}{theorem}{thmMCMub}
\label{thm:matchingUB}
For the $\MCM$ problem in the edge/vertex-arrival model, the $\LGreedy$ algorithm is $t$-competitive for any $1<t<2$ and incurs at most $\frac{(2-t^*)}{(t^*-1)(3-t^*)}+\frac{t^*-1}{3-t^*}$ amortized recourse, where $t^*$ is the largest number such that $t^* \leq t$ and $t^*= 1 +\frac{1}{j}$ for some integer $j$.
\end{restatable}

\begin{proof} \textbf{(Ideas.)}
Consider the connected components generated by the union of edges chosen by $\LGreedy$ or by $\opt$. Let $C_i$ be the components in the graph and $\tr_i$ be the total recourse incurred by the elements in $C_i$ (from very beginning till the end), the amortized recourse given by the whole graph will be upper-bounded by $\max_i \frac{\tr_i}{|C_i|}$ (Observation~\ref{obs:sum-max}). 

By selecting $L = \lceil\frac{2-t}{t-1}\rceil$, the path eliminations only happen at odd-size components with length from $3$ to $2(\lceil \frac{1}{t-1}\rceil-1)+1$ (note that $\lceil \frac{1}{t-1}\rceil \geq 2$ since $1<t<2$). 
Moreover, for such a $(2k+1)$-edge augmenting path, the total recourse incurred by the $2k+1$ elements in the path is at most $1+\sum_{k=1}^{\lceil\frac{1}{t-1}\rceil-1} 2k$.
Hence, the amount of amortized recourse incurred by this component is at most $\frac{(\lceil\frac{1}{t-1}\rceil-1)\cdot\lceil\frac{1}{t-1}\rceil+1}{2\lceil\frac{1}{t-1}\rceil-1}$.

If $t = 1 + \frac{1}{j}$ for some integer~$j$, the amortized recourse for a component is at most 
$\frac{(\lceil\frac{1}{t-1}\rceil-1)\cdot\lceil\frac{1}{t-1}\rceil+1}{2\lceil\frac{1}{t-1}\rceil-1} =\frac{2-t}{(t-1)(3-t)} + \frac{t-1}{3-t}$.
It can be adapted to the case in which there is no integer $j$ such that $t = 1 + \frac{1}{j}$ by rounding down $t$ to the largest $t^*\leq t$ such that $t^* = 1+\frac{1}{j}$ for some integer $j$. 
By eliminating all augmenting paths that have length at most $\frac{2}{t^*-1}-1$, the amount of incurred amortized recourse is at most $\frac{2-t^*}{(t^*-1)(3-t^*)} + \frac{t^*-1}{3-t^*}$, and the algorithm attains a competitive ratio of $t^* \leq t$.
\end{proof}

\begin{restatable}{theorem}{thmMCMlb}
No deterministic $t$-competitive online algorithm can incur amortized recourse less than $\frac{(2-t^*)}{(t^*-1)(3-t^*)}$ in the worst case.
\end{restatable}
\hide{
\hhl{
\begin{proof} \textbf{(Ideas.)}
If $\frac{n+2}{n+1} \leq t < \frac{n+1}{n}$ for some integer $n \geq 1$, we release a sequence of $2n+1$ edges that form a path.

Consider any $1 \leq k \leq n$, the following invariants hold for any $t$-competitive algorithm (the full proof for the invariants is in Appendix):
\begin{enumerate}[label={(\bfseries I\arabic*)}]
\setlength\itemsep{0em}
    \item For a path with length $2k+1$, a $t$-competitive algorithm has to accept $k+1$ edges.
    \item For a path with length $2k$, a $t$-competitive algorithm has to accept $k$ edges.
    \item When an instance is increased from a $2(k-1)+1$ path to a $2k+1$ path, a $t$-competitive algorithm incurs at least $2k$ amount of recourse.
\end{enumerate}

Given the invariants I1, I2, and I3, the $(2n+1)$-length path instance incurs recourse with total amount at least $\sum_{k=1}^n (2k) = n\cdot(n+1)$. Therefore, any $t$-competitive algorithm incurs at least $\frac{n\cdot(n+1)}{2n+1}$ amortized recourse for this instance. Let $t^* = \frac{n+2}{n+1} \leq t$. It follows that $n = \frac{2-t^*}{t^*-1}$. 
Therefore, the amortized recourse is at least $\frac{n\cdot(n+1)}{2n+1} = \frac{(2-t^*)}{(t^*-1)(3-t^*)}$.
\end{proof}
}
}

\section{Minimum Vertex Cover}\label{sec:VC}

In this section, we propose a special version of the $\tas_t$ algorithm, $\texttt{Duo-Halve}$, that attains a competitive ratio of $2-\frac{2}{\opt}$ for the $\MVC$ problem with optimal vertex cover size $\opt$ in polynomial time.
The \texttt{Duo-Halve} algorithm uses an optimal solution as the yardstick with $t= 2-\frac{2}{\opt}$. However, the computation of the optimal solution of $\VC$ is very expensive. Thus, we maintain a maximal matching greedily (as the well-known $2$-approximation algorithm for $\VC$) on the current input graph and only select vertices that are saturated by the matching.
If the $\PEA$ algorithm rejects two of these vertices, the competitive ratio is at most $2-\frac{2}{\opt}$.
We show that there always exists a feasible solution where the last two matched edges only contribute two vertices to the solution, or the optimal solution has size at least $|M|+1$, where $M$ is the current maximal matching, and the \texttt{Duo-Halve} algorithm is therefore $(2-\frac{2}{\opt})$-competitive.

\hide{During the process, the algorithm maintains a maximal matching $M(\ins)$ on the current input graph $\ins$ to construct a solution $\PEA(\ins)$ (we omit the parameter $\ins$ when the context is clear). 
The algorithm only selects vertices that are saturated by the matching and rejects as many vertices as possible from the two latest matched edges. 
It is clear that if we reject $2$ such vertices, the competitive ratio of the algorithm is at most $2-\frac{2}{\opt}$ since $\opt \geq |M|$. 
We show that if we cannot reject $2$ such vertices, the optimal solution size must be big so the competitive ratio is still $2-\frac{2}{\opt}$ (Theorem~\ref{thm:VCCompRatio}).}

In the following discussion, we use some terminology. 
Let $\me1$ and $\me2$ be the most and the second-most recently matched edges respectively. Also, let $\vmi$ be the vertices saturated by the maximal matching $M(\ins)$.
The $\pea$ algorithm partitions the vertices into three groups: \textbf{Group-1}: the endpoints of $\me1$ or $\me2$, \textbf{Group-2}: the vertices in $\vm$ but not in Group-1, and \textbf{Group-3}: the vertices in $V\setminus\vm$.
\hide{
A matched edge is \emph{full} if both of its endpoints are selected by the algorithm $\pea$. Otherwise, the edge is \emph{half}. 
The algorithm \emph{halves} a matched edge $(u,v)$ by producing a valid vertex cover while only accepting either $u$ or $v$.
A matched edge $(u,v)$ is \emph{halvable} if there exists a vertex cover that contains exactly one of $u$ and $v$. 
A half edge \emph{flips} if the accept/reject status of its endpoints is swapped.
A \emph{configuration} of a set of edges is a set of accept/reject statuses associated with each endpoint of those edges.
}

\runtitle{\texttt{Duo-Halve} Algorithm ($\PEA$).} 
When a new vertex $v$ arrives, if an edge $(p, v)$ is added to $M(\ins)$, then it introduces a new $\me1$ (namely $(p, v)$). The algorithm first accepts all \textbf{Group-2} vertices that are adjacent to $v$. 
Then, the algorithm decides the assignment of $\me1$ and $\me2$ and minimizes the number of accepted endpoints of $\me1$ and $\me2$.
If there is a tie, we apply the one that accepts fewer endpoints in $\me1$ and/or incurs less recourse.
\hide{
by testing if they can be both halved by the \texttt{HalveBoth procedure} as follows.

If $\me1$ is half or one of its endpoints is $v$, the $\PEA$ algorithm halves $\me2$ if it is valid giving the current configuration of $\me1$ or $v$ being accepted. 
Otherwise, $\PEA$ halves $\me2$ if it is valid by flipping $\me1$ or late-accepting $p$. 
Otherwise, in the vertex cover returned by $\PEA$, $\me1$ is full and $\me2$ is unchanged. (See Figure~\ref{fig:flow} in Appendix for a detailed flow diagram about this \texttt{HalveBoth} procedure.)  
}




The $\PEA$ algorithm returns a feasible solution in $O(n^3)$ time, where $n$ is the number of vertices in the graph.
Intuitively, the algorithm maintains a valid solution as it greedily covers edges using vertices in the maximal matching, with the exception of $\me1$ and $\me2$, where it carefully ensures that a feasible configuration is chosen.
Furthermore, the most computationally-expensive component of the $\PEA$ algorithm, which checks the validity of a constant number of configurations by looking at the neighborhoods of $\me1$ and $\me2$, runs in $O(n^2)$ time for each new element.
\hide{
\begin{restatable}{lemma}{lmVCtime}
\label{lm:vc-time}
The $\PEA$ algorithm always returns a valid vertex cover in $O(n^3)$ time, where $n$ is the number of vertices in the graph.
\end{restatable}
}

\hide{
\begin{observation}\label{obs:subset}
If $\opt \geq |M|+1$, then $\frac{\PEA}{\opt} \leq 2-\frac{2}{\opt}$.
\end{observation}
\hide{\begin{proof}
Since the $\PEA$ algorithm only accepts the vertices saturated by $M$, $\PEA \le 2|M|$. Therefore, $\frac{\PEA}{\opt} \le \frac{2|M|}{|M| + 1} = \frac{2|M|+2-2}{|M|+1}$ $= 2-\frac{2}{|M|+1}\leq 2-\frac{2}{\opt}$.
\end{proof}}

The above observation allows us to show another condition under which $\opt \ge |M| + 1$ and thus our desired bound of $2 - \frac{2}{\opt}$ on the competitive ratio is achieved.
The contrapositive of this lemma also provides us with a necessary characteristic of any ``problematic'' instance.

\begin{restatable}{lemma}{lmVCsubset}
\label{lm:subset}
If $\opt \setminus \vm \neq \varnothing$ or there exists an edge $(u,v)$ in $M$ such that $\{u,v\}\subseteq \opt$, then $\frac{\PEA}{\opt} \leq 2-\frac{2}{\opt}$.
\end{restatable}

\begin{corollary}\label{col:subset}
If $2-\frac{2}{\opt} < \frac{\PEA}{\opt}$, then $\opt \subseteq \vm$.
\end{corollary}
\hide{
\begin{proof}
This corollary follows directly as the contrapositive to Lemma~\ref{lm:subset}.
\end{proof}
}

We then present another observation, which helps with our analysis of $\PEA$ by providing a feasibility guarantee for any newly-revealed $\me1$.
\begin{observation}\label{obs:shift-me1}
Immediately after an edge is added to the matching, $\me1$ is halvable by accepting the newly-revealed vertex.
\end{observation}
}
\hide{\begin{proof}
Let $\me1 = (p, v)$, where $v$ is the newly-revealed vertex.
Prior to $v$ being revealed, $p$ is unmatched and thus rejected.
Furthermore, accepting $v$ covers all newly-revealed edges.
Therefore, $\me1$ is halvable by accepting $v$ and rejecting $p$.
\end{proof}}


We first show that if $\PEA$ fails to produce a solution where it accepts only one vertex of $\me1$, then $\opt \geq |M|+1$.
The intuition is that if $\PEA$ has to accept both endpoints of $\me1$, there must be at least one \textbf{Group-3} vertex in each of the endpoints' neighborhoods. Therefore, the optimal solution has to cover the corresponding edges with at least two vertices.  



\begin{restatable}{lemma}{lmVCfull}
\label{lm:full-me1_new}
In the assignment of $\pea$, if both endpoints of $\me1$ are selected, then the optimal solution must contain at least two vertices in $\me1 \cup (V \setminus \vm)$, and $\frac{\pea}{\opt} \leq 2-\frac{2}{\opt}$.
\end{restatable}
\vspace{-2ex}
\hide{
\begin{corollary}\label{col:full-me1}
In the assignment of $\pea$, if $\me1$ is full, then $ \frac{\pea}{\opt} \leq 2-\frac{2}{\opt}$.
\end{corollary}
}
\hide{
\begin{proof}
By Lemma~\ref{lm:full-me1} and Observation~\ref{obs:subset}.
\end{proof}
}
\hide{After taking care of the above preliminaries, we show that our desired bound of $2 - \frac{2}{\opt}$ on the competitiveness of the $\PEA$ algorithm is indeed attainable.
}
\begin{restatable}{theorem}{VCcr}
\label{thm:VCCompRatio_new}
The $\PEA$ algorithm is $(2-\frac{2}{\opt})$-competitive.
\end{restatable}
\begin{proof} \textbf{(Ideas.)}
In any possible solution provided by $\PEA$, there are three states based on the configuration of $\me1$ and $\me2$:
1) both $\me1$ and $\me2$ are half, 2) $\me1$ is half and $\me2$ is full, and 3) $\me1$ is full.
In state $1$, we can directly show that the bound holds since $\PEA \leq 2|M|-2$. The bound holds for state $3$ by Lemma~\ref{lm:full-me1_new}.

State $2$ requires more involved analysis.
If an endpoint of $\me2$ has a rejected \textbf{Group-2} neighbor, then $\PEA$ rejects at least two vertices in $V_M$ (this \textbf{Group-2} neighbor and $1$ $\me1$ vertex) and $\PEA \le 2|M| - 2$. 
Otherwise, if at least one endpoint of $\me2$ has no \textbf{Group-3} neighbor, then we can show that there is no solution based on the maximal matching containing only 2 \textbf{Group-1} vertices. 
This means that $\opt$ must contain either a \textbf{Group-3} vertex or 3 \textbf{Group-1} vertices, and thus $\opt \ge |M| + 1$.
Finally, if each endpoint of $\me2$ has a \textbf{Group-3} neighbor, then $\opt$ must either select a \textbf{Group-3} vertex or both endpoints of $\me2$, and $\opt \ge |M| + 1$.
\end{proof}

\hide{With the upper bound on the competitive ratio of the $\PEA$ algorithm proven, we then show an upper bound on the amortized recourse incurred by $\PEA$.
Prior to doing so, we show a restriction on the possible changes to the $\PEA$ algorithm's solution when a new edge is added to the maximal matching $M$.
By extension, this limits the possible number of late operations incurred in this scenario.

\begin{restatable}{lemma}{lmVCshift}
\label{lm:shift-me2}
When an edge shifts from $\me1$ to $\me2$, it either goes from full to half or remains unchanged.
\end{restatable}

In particular, Lemma~\ref{lm:shift-me2} tells us that a half edge shifting from $\me1$ to $\me2$ will not incur any late operations (i.e. it won't become full and it won't flip).

In the worst case, there can be multiple late accepts on \textbf{Group-2} vertices
, and $4$ late operations on vertices in $\me1$ and $\me2$. In the following theorem, we use a potential function to prove that the amortized late operations per vertex is at most $\frac{10}{3}$.
}
For a single newly-revealed vertex, the amount of recourse incurred can be up to $O(n)$.
Even if we restrict our consideration to $\me1$ and $\me2$, a single new vertex can incur recourse at most~$4$.
However, this cannot happen at every input.
We use a potential function to show that the amortized recourse incurred by $\PEA$ is at most $3.33$.

\begin{restatable}{theorem}{VCar}
\label{thm:VCUB}
The amortized recourse incurred by $\PEA$ is at most $\frac{10}{3}$.
\end{restatable}
\begin{proof} \textbf{(Ideas.)}
We prove the theorem by using a potential function.
To this end, we define an edge $(u, v)$ as being \emph{free} if there exist feasible assignments both by either accepting $u$ or by accepting $v$.
Also, we define a matched edge with only one endpoint selected as being \emph{expired} if it is neither $\me1$ nor $\me2$.
Finally, we define $A$ as the set of vertices accepted by $\PEA$.
Using these terms, we define the potential function $\Phi$ as
\vspace{-0.5em}
\begin{equation*}
\Phi := |\{(u, v) | (u, v) \mbox{ expired}\}| + \frac{1}{3}|A \cap (\me1 \cup \me2)| + \frac{2}{3} \cdot \mathds{1}[ \me2 \mbox{ is free}]
\end{equation*}
Furthermore, at any given moment in the input sequence where the matching constructed by $\PEA$ contains at least $2$ edges, the status of $\me1$ and $\me2$ is characterized by one of $6$ states according to their possible combinations of selection statuses of their endpoints.
We also differentiate between the two half possibilities for $\me1$, since the newly-revealed vertex in $\me1$ can be accepted without incurring a late operation when there is a new $\me1$.

We show that, for any possible state transition triggered by a newly-revealed vertex, the number of incurred late operations $\lo$ added to the change in potential $\dphi$ is bounded above by $\frac{10}{3}$.
Note that, for any newly-revealed vertex $v$, $v$ may be adjacent to $k \ge 0$ rejected vertices that are matched by some expired edge.
This incurs $k$ late operations, but also decreases $\Phi$ by $k$, so this may be ignored when computing $\lo + \dphi$.
Since $\Phi_0 = 0$ and $\Phi_i \ge 0$, this allows us to conclude the statement of our theorem.
\end{proof}

Moreover, we can show a lower bound by constructing a family of instances that alternates between incurring a late accept on a \textbf{Group-2} vertex, and $4$ late operations on $\me1$ and $\me2$.
(See Figure~\ref{fig:vc-ar-lb} in Appendix.)

\begin{restatable}{lemma}{VCARLB}
For any $\varepsilon > 0$, there exists an instance such that $\PEA$ incurs amortized recourse strictly greater than $\frac{5}{2} - \varepsilon$.
\end{restatable}
\hide{
\begin{proof}
Consider $\PEA$ against an instance that reveals vertices as follows:
\begin{enumerate}
    \item isolated vertex $a$. $\PEA$ rejects $a$.
    \item $b$ adjacent to $a$. $\PEA$ adds $(a, b)$ to the matching, and accepts $b$.
    \item $c$ adjacent to $a$. $\PEA$ late-accepts $a$, and late-rejects $b$.
    \item $d$ adjacent to $c$. $\PEA$ adds $(c, d)$ to the matching, and accepts $d$.
    \item $e$ adjacent to $b$ and $c$. $\PEA$ late-accepts $b$ and $c$, and late-rejects $a$ and $d$.
    \item $f$ adjacent to $e$ and $a$. $\PEA$ adds $(e, f)$ to the matching, accepts $f$, and late-accepts $a$.
\end{enumerate}
This instance is illustrated in Figure~\ref{fig:vc-ar-lb} (see appendix).
Observe that the configuration of $\{c, d, e, f\}$ after $f$ is revealed is isomorphic to the configuration of $\{a, b, c, d\}$ after $d$ is revealed. Furthermore, all edges with an endpoint outside of $\{c, d, e, f\}$ are covered by a vertex outside of $\{c, d, e, f\}$.
Thus, this instance can be extended indefinitely by repeating steps $5$ and $6$.
Since these two steps together incur $5$ late operations by revealing two vertices, the family of instances described here incurs asymptotic amortized recourse $\frac{5}{2}$.
\end{proof}
}

Finally, we show that the analysis in Theorem~\ref{thm:VCCompRatio_new} is tight 
for a class of online algorithms where its solution only contains vertices saturated by the matching maintained throughout the process in an incremental manner. In other words, 
no online algorithm in this class achieves a lower competitive ratio, no matter how much amortized recourse it uses.
This is done by using an adversary that constructs an arbitrary number of triangles that all share a common vertex $v$.
This common vertex is revealed last, so all edges not incident to $v$ are added to the matching and $v$ is rejected. (See Figure~\ref{fig:vc-cr-lb} in Appendix.)

\hide{
\begin{definition}
An algorithm for vertex cover is \emph{incremental matching-based} if it maintains a maximal matching throughout the process in an incremental manner, and its solution only contains vertices saturated by the matching.
\end{definition}
}

\begin{restatable}{theorem}{VROPT}
No deterministic incremental matching-based algorithm achieves a competitive ratio smaller than $2-\frac{2}{\opt}$.
\end{restatable}
\hide{
\begin{proof}
Consider any incremental matching-based algorithm against an instance where $k$ disconnected edges are revealed via their endpoints.
By nature of the incremental matching-based algorithm, each of these $k$ edges will be added to the matching, and at least one vertex from each pair will be accepted.
Then, a final vertex is revealed, adjacent to all previously-revealed vertices.
The incremental matching-based algorithm will not accept this vertex, as it is not matched, but must accept all other vertices for a vertex cover of size $2k$.
However, the optimal solution consists of the last revealed vertex, and one vertex from each pair.
Thus, against this instance, any incremental matching-based algorithm will be $\frac{2k}{k+1} = 2 - \frac{2}{\opt}$ times worse than the optimal solution.
Therefore, no incremental matching-based algorithm can achieve a competitive ratio smaller than $2 - \frac{2}{\opt}$.
\end{proof}
}


\bibliographystyle{plainurl}
\bibliography{references}

\newpage
\appendix
\section{Algorithms and Figures}
\begin{algorithm}
\caption{$\tas_t$ algorithm for monotone-sum graph problems}\label{alg:tas}
\begin{algorithmic}
\State $\alg \gets 0$
\While{new element $v$ arrives}
    \State $g \gets$ the best value from $[0,w_\text{max}]$ such that no feasibility constraint is violated
    \If{the new assignment will fail to be $t$-competitive} \Comment{$\max\{\frac{\alg+g}{\opt_{}}, \frac{\opt_{}}{\alg+g}\} > t$}
        \State \textsc{Switch}($\opt$)
    \Else
        \State incorporate the greedy assignment
    \EndIf
    \State $\alg \gets$ the value of $\tas_t$'s current assignment
\EndWhile
\\
\Function{\textsc{Switch}}{assignment $A$}
    \For{every element $x$}
        \If{$\tas_t(x) \neq A(x)$}
            \State change the assignment of element $x$ into $A(x)$
        \EndIf
    \EndFor
\EndFunction
\end{algorithmic}
\end{algorithm}

\bigskip

\begin{algorithm}
\caption{\texttt{Duo-Halve} algorithm ($\PEA$) for Minimum Vertex Cover Problem}\label{alg:VC2}
\begin{algorithmic}
\State $\me1 \gets \varnothing$, $\me2 \gets \varnothing$, $\vm \gets \varnothing$
\While{new vertex $v$ arrives}
    \If{there is a vertex $p \in N(v)\cup (V\setminus \vm)$} 
        \State $\me2 \gets \me1$
        \State $\me1 \gets (p,v)$ \Comment{$(p,v)$ is a new matched edge. If there is more than one $p$, choose one arbitrarily.}
        \State add $p$ and $v$ into $\vm$
        \State \texttt{LateAccept} all rejected vertices in $(\vm\setminus\{$vertices in $\me1$ or $\me2\})\cap N(v)$
        \State \texttt{HalveBoth}($\me1, \me2$)
    \Else 
        \State \texttt{LateAccept} all rejected vertices in $\vm\cap N(v)$
        \State \texttt{HalveBoth}($\me1, \me2$)
    \EndIf
\EndWhile
\\
\State \textbf{Function }{\texttt{HalveBoth}}(matched edge $\me1$, matched edge $\me2$)
    \State Among accept/reject configurations of $\me1$ and $\me2$ that yield a valid vertex cover, return one that maximizes the number of half edges among $\me1$ and $\me2$ with the minimum number of late operations (see Figure~\ref{fig:flow} for details)
    
\State \textbf{end Function}
\end{algorithmic}
\end{algorithm}

\begin{figure}[h]
    \centering
    \includegraphics[scale=0.58]{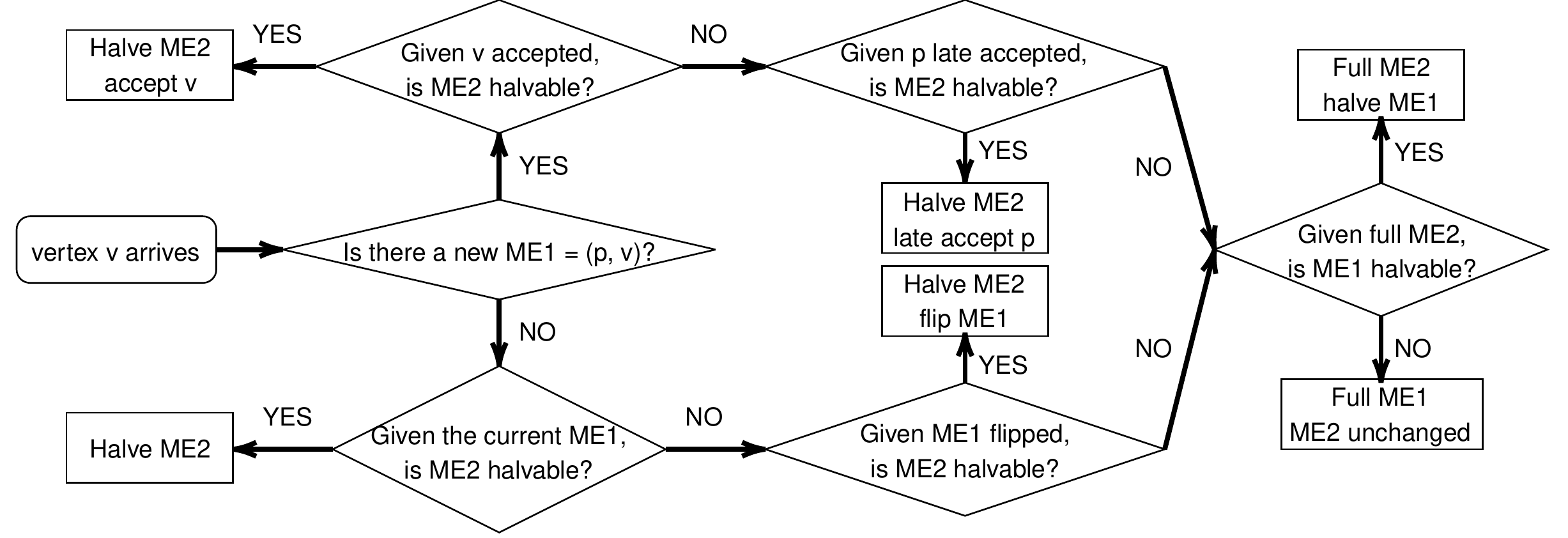}
    \caption{An illustration of the flow of \texttt{HalveBoth}($\me1,\me2$).}
    \label{fig:flow}
\end{figure}

\begin{figure}
\centering
    \includegraphics[scale=0.9]{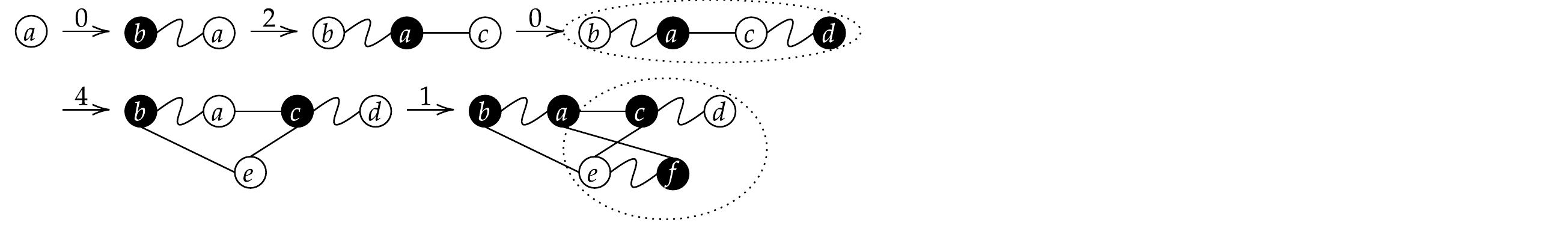}
    \caption{Adversarial instance for $\VC$ such that $\PEA$ incurs asymptotic amortized recourse $\frac{5}{2}$. Each arrow's number denotes the number of late operations incurred by the next vertex's reveal. The dotted ovals highlight the repeating structure.}
    \label{fig:vc-ar-lb}
\end{figure}

\begin{figure}
\centering
    \includegraphics[scale=0.5]{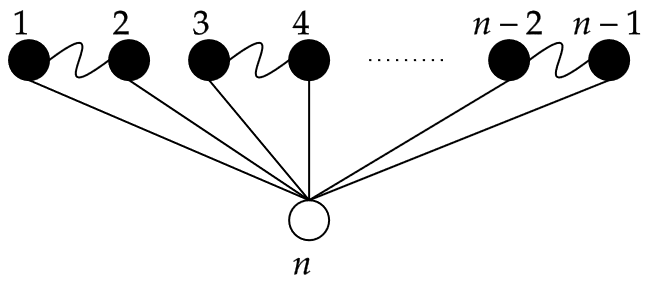}
    \caption{Adversarial instance for $\VC$ such that any incremental matching-based algorithm is exactly $(2 - \frac{2}{\opt})$-competitive. Vertices are labeled by their release order. Any such algorithm must accept $n-1$ vertices, whereas the optimal solution contains $\frac{n-1}{2} + 1$ vertices.}
    \label{fig:vc-cr-lb}
\end{figure}

\newpage

\section{General Algorithm for Monotone-Sum Graph Problems Full Proof}
\theoremGeneralUB*
\begin{proof}
First, note that in the $\tas_t$ algorithm, recourse is incurred only at the moments when a switch happens. We partition the process of the algorithm into \emph{phases} according to the switches. Phase $i$ consists all the events after the $(i-1)$-th switch until the $i$-th switch. Let~$\ins_i$ denote the set of elements that are released during the $i$-th phase and can be assigned non-zero value. Let $\tr(\ins_i)$ and $\ar(\ins_i)$ denote the total recourse and the amortized recourse incurred at the $i$-th switch, respectively. The total amortized recourse is given by $\ar(\ins) = \frac{\sum_{i} \tr_i}{\sum_i|\ins_i|} \leq \max_i \frac{\tr_i}{|\ins_i|}$ (Observation~\ref{obs:sum-max}). In the following, we prove that for any switch $i$, $\frac{\tr_i}{|\ins_i|} \leq \frac{w_\text{max}\cdot(t+1)}{t-1}$ and thus $\ar(\ins) \leq \frac{w_\text{max}\cdot(t+1)}{t-1}$.

We first show that any optimal schedule is incremental. In this case, the optimal solution can be seen as an incremental $1$-approximation algorithm. Given any graph $G$ and its subgraph~$H$, since the problem is monotone, $\opt(H) \leq \opt(G)$. 
Moreover, consider the projection $\opt^\prime$ of $\opt(G)$ on $H$, it is a feasible solution of $H$ since the problem is monotone. Since $\opt^\prime$ is feasible, $value(\opt(G)) \geq value(\opt^\prime)$ for maximization problems, and $value(\opt(G)) \leq value(\opt^\prime)$ for minimization problems.
That is, the projection $\opt^\prime$ of $\opt(G)$ on $H$ does not have a better objective value than the assignment $\opt(H)$.

Let $\ralg$ be the yardstick incremental $\alpha$-approximation algorithm referenced (where $\alpha=1$ if we reference the optimal exact solution).
Let $\ralg_i$ and $\tas_i$ denote the value of the yardstick solution and the value of $\tas_t$'s solution right after the $i$-th switch, respectively. By the $\tas_t$ algorithm, $\tas_i = \ralg_i$. 
Consider the input $x$ that triggers the $(i+1)$-th switch. Let $\alg$ denote the value of $\tas_t$ algorithm right before the arrival of $x$. Once $x$ is revealed, the $\tas_t$ algorithm attempts to greedily assign a value $g \in [0,w_\text{max}]$ to $x$, while the value of the $\ralg$ assignment for the whole instance is $\ralg_{i+1}$. By the definition of the $\tas_t$ algorithm, the $(i+1)$-th switch is triggered since $\max\{\frac{\alg+g}{\ralg_{i+1}}, \frac{\ralg_{i+1}}{\alg+g}\} > t$. In the worst case, the $\tas_t$ algorithm changes the decision on every element. Since we check the $t$-competitiveness before assigning the value greedily, $\tr_{i+1} \leq \alg + \ralg_{i+1}$.

\smallskip

\runtitle{Maximization problems.} 
According to the condition that triggers the switching, at the moment when the $(i+1)$-th switch happens, $\ralg_{i+1} > t\cdot(\alg+g)$. Hence, the total recourse $\tr_{i+1} \leq \alg+\ralg_{i+1} < (1+1/t) \cdot \ralg_{i+1}$. 

Now, we bound the number of elements released between switch $i$ and switch $i+1$ from below. Since $\ralg$ is incremental, and the assigned values are non-negative, the adversary needs to release a sufficient number of elements such that the value of the referenced assignment increases enough to trigger the switch.
Let $\ralg^\prime$ be the projection of $\ralg_{i+1}$ on the graph at the beginning of the $i+1$-th phase.
That is, $\ralg^\prime = \sum_{j=1}^i \sum_{x \text{ released in phase }j} \ralg_{i+1}(x)$. Since the $\ralg$ algorithm is incremental, $\ralg^\prime \leq \ralg_i$.
Since every new element brings at most $w_\text{max}$ additional cost in the $\ralg$ assignment value, the minimum number of elements that the adversary has to release such that the $\ralg$ value increases enough is at least $\frac{\ralg_{i+1}-\ralg^\prime}{w_\text{max}}$. Therefore, $|\ins_{i+1}| \geq \frac{\ralg_{i+1}-\ralg^\prime}{w_\text{max}} \geq \frac{\ralg_{i+1}-\ralg_{i}}{w_\text{max}}$. By the $\tas_t$ algorithm and the fact that the assigned values are non-negative, $|\ins_{i+1}| \geq \frac{\ralg_{i+1}-\ralg_{i}}{w_\text{max}} = \frac{\ralg_{i+1}-\alg_i}{w_\text{max}} \geq \frac{\ralg_{i+1}-\alg}{w_\text{max}}$. By the switch triggering condition, $\alg+g < \ralg_{i+1}/t$. Hence, $|\ins_{i+1}| \geq \frac{\ralg_{i+1}-(\alg+g)}{w_\text{max}} \geq \frac{\ralg_{i+1}-\ralg_{i+1}/t}{w_\text{max}}$.

By definition, $\ar_{i+1} = \frac{\tr_{i+1}}{|\ins_{i+1}|} \leq \frac{w_\text{max}\cdot(1+1/t)\cdot \ralg_{i+1}}{(1-1/t)\cdot\ralg_{i+1}} = \frac{w_\text{max}\cdot(t+1)}{t-1}$.

\smallskip

\runtitle{Minimization problems.} 
The $(i+1)$-th switch occurs when $\alg + g \geq t\cdot \ralg_{i+1}$. Hence, $\tr_{i+1} \leq \alg + \ralg_{i+1} \leq \alg + (1/t)\cdot (\alg+g) \leq (1+1/t) \cdot \alg + w_\text{max}/t$.

Now, we bound the number of elements released between switch $i$ and switch $i+1$ from below. Unlike in maximization problems, in minimization problems it is possible that $\ralg_{i+1} = \ralg_{i}$ but the assignment at the $(i+1)$-th switch is completely different from that at the $i$-th switch and causes massive amount of recourse for $\tas_t$ to remain $t$-competitive. Therefore, we cannot bound the number of new elements by the change of $\ralg$ value but instead do so according to the change of $\tas_t$ assignment's total value. 
Right before the input that triggers the $(i+1)$-th switch, $\tas_t$ has value $\alg \geq \alg_i$ since the assigned values are non-negative (note that $\alg=\alg_i$ if it is the input that triggers the $i$-th switch, and the switches happen in a row). Together with the reveal of the input that triggers the $(i+1)$-th switch, $|\ins_{i+1}|$ is at least $\frac{\alg-\tas_i}{w_\text{max}}+1$.
Therefore, $|\ins_{i+1}| \geq \frac{\alg-\tas_i}{w_\text{max}}+1 = \frac{\alg-\ralg_i}{w_\text{max}}+1 \geq \frac{\alg-\ralg_{i+1}}{w_\text{max}}+1 > \frac{\alg-(1/t)\cdot(\alg+g)}{w_\text{max}}+1 \geq \frac{\alg-(1/t)\cdot(\alg+w_\text{max})}{w_\text{max}}+1 = \frac{(1-1/t)\cdot\alg + w_\text{max}\cdot(1-1/t)}{w_\text{max}}$.

By definition, $\ar_i = \frac{\tr_{i+1}}{|\ins_{i+1}|} \leq w_\text{max}\cdot \frac{(1+1/t) \cdot \alg + w_\text{max}/t}{(1-1/t)\cdot\alg + w_\text{max}\cdot(1-1/t)} = w_\text{max}\cdot\frac{(t+1)\cdot\alg + w_\text{max}}{(t-1)\cdot \alg + w_\text{max}\cdot(t-1)} = \frac{w_\text{max}}{t-1}\cdot\frac{(t+1)\cdot \alg + w_\text{max}}{\alg + w_\text{max}} \leq \frac{w_\text{max}}{t-1}\cdot(t+1) = \frac{w_\text{max}\cdot(t+1)}{t-1}$. 
\end{proof}

\bigskip

\corFractionalUB*
\begin{proof}
The proof is similar to the one of Theorem~\ref{thm:general_UB}. 
The new ingredient of the proof is that, given the same amount of value changed, the number of elements whose assigned values are changed is upper-bounded by $\frac{\text{total amount of value change value}}{w_\text{min}} \leq \frac{\alg+\ralg_{i+1}}{w_\text{min}}$. 
The other arguments in the proof of Theorem~\ref{thm:general_UB} still follow. 
Hence, the amortized recourse is at most $\frac{w_\text{max}\cdot(t+1)}{w_\text{min}\cdot(t-1)} = \frac{t+1}{w_\text{min}\cdot(t-1)}$.
\end{proof}

\section{Independent Set Full Proof}
\lmISreduc*
\begin{proof}
For any instance $(G, \sigma)$ of $\IS$, any vertex $v \in G$ released after switch $i$ and before switch $i+1$ belongs to one of the following four types:
\begin{itemize}
	\item $v$ is selectable by $\tas_t$, and $v \in \opt_{i+1}$.
	\item $v$ is selectable by $\tas_t$, and $v \notin \opt_{i+1}$.
	\item $v$ is not selectable by $\tas_t$, and $v \in \opt_{i+1}$.
    \item $v$ is not selectable by $\tas_t$, and $v \notin \opt_{i+1}$.
\end{itemize}
Since $\tas_t$\ accepts any newly-revealed vertex that it can, any vertex selectable by $\tas_t$\ will be accepted.
\\
Then, consider any vertex $v$ of the fourth type. Such a vertex does not incur any late operation at switch $i + 1$, nor does it have any impact on either $\tas_t$'s solution or $\opt_{i+1}$. Thus, an instance $(G, \sigma')$ such that $\sigma'$ releases $v$ immediately after switch $i + 1$ and is otherwise identical to $\sigma$ will incur no less amortized recourse than $(G, \sigma)$. Furthermore, if $v$ would be released in $\sigma'$ after the last switch, it can instead be removed from $G$. This transformed instance $(G \setminus \{v\}, \sigma' \setminus \{v\})$ incurs the same recourse as $(G, \sigma)$ and contains one less vertex, and thus incurs more amortized recourse. In this way, we can transform any general instance $(G, \sigma)$ into an instance $(G', \sigma')$ containing no vertices of the fourth type such that the amortized recourse incurred by $(G', \sigma')$ is no less than that incurred by $(G, \sigma)$.
\\	
Furthermore, consider any vertex $v$ of the first type in an instance $(G', \sigma')$ containing no vertices of the fourth type. Such a vertex does not incur any late operation at switch $i + 1$, and decreases the ratio between the size of the optimal solution and $\tas_t$'s solution. Thus, $v$ cannot directly contribute towards an increase in amortized recourse incurred by switch $i + 1$. However, it may have an indirect impact on $\tas_t$'s solution. In particular, there may be some neighbor $u$ of $v$ that would be in $\tas_t$'s solution if $v$ were not present in the graph. Since $\tas_t$\ accepts any selectable vertex, $v$ must be released before $u$ and thus $u$ is not selectable by $\tas_t$\ when it is released. Therefore, since $u$ is not part of $\opt_{i+1}$, $u$ is a vertex of the fourth type. However, $(G', \sigma')$ contains no vertices of the fourth type, so there is no such vertex $u$. Then, using the same approach as for vertices of the fourth type, we can transform any instance $(G', \sigma')$ containing no vertices of the fourth type into an instance $(G'', \sigma'')$ containing no vertices of the first or fourth types.
\end{proof}

\bigskip

\thmISub*
\begin{proof}
We show that, for any reduced instance from Lemma~\ref{lm:reduction}, $\tas_t$ will incur at most $\frac{t}{t-1}$ amortized recourse, and thus that this upper bound holds for any instance. To do this, we show that for any switch~$i$, $\frac{\tr_i}{|\ins_i|} \leq \frac{t}{t-1}$ and thus $\ar(\ins) = \frac{\sum_i \tr_i}{\sum_i |\ins_i|} \leq \max_i\frac{\tr_i}{|\ins_i|} \leq \frac{t}{t-1}$.
\\
Consider a scheme in which each newly-revealed vertex carries budget $B$, and the vertices revealed between switch $i$ and switch $i+1$ must pay the full cost of the recourse incurred by switch $i+1$. Then, the total budget carried by these newly-revealed vertices must be at least $ALG + \opt_{i+1}$.
\\
Furthermore, the total available budget at switch $i+1$ is at least $B$ times the number of vertices revealed between switches $i$ and $i+1$, which in turn is at least $B$ times the number of these vertices of the third type. Note that~$\opt_{i+1}$ can always be constructed by extending an existing independent set of size at least $\opt_{i-1}$, so the number of vertices of the third type is bounded above by $\opt_{i+1} - \opt_{i-1}$. Therefore, to have sufficient budget,~$B$ must satisfy~$B \ge \frac{ALG + \opt_{i+1}}{\opt_{i+1} - \opt_{i-1}}$.
\\
Denote $k$ the number of vertices of the second type revealed between switches~$i$ and~$i+1$, and $k'$ the number of vertices of the second type revealed between switches $i-1$ and $i$. Then $ALG = \opt_i + k$ and $ALG' = \opt_{i-1} + k'$, so $\opt_{i+1} \ge t(\opt_i + k)$, $\opt_{i+1} - 1 < t \cdot ALG$ which is equivalent to $\opt_{i+1} < t \cdot \opt_i + tk + 1$, and $\opt_i - 1 < t \cdot ALG'$ which is equivalent to $\opt_{i-1} > \frac{\opt_i}{t} - k' - \frac{1}{t}$.
\\
Therefore, we have $B \ge \frac{ALG + \opt_{i+1}}{\opt_{i+1} - \opt_{i-1}} \ge \frac{\opt_i + k + t(\opt_i + k)}{t \cdot \opt_i + tk + 1 - \frac{\opt_i}{t} + k' + \frac{1}{t}} = \frac{\opt_i(t + 1) + tk + k}{\opt_i(t - \frac{1}{t}) + tk + 1 + \frac{1}{t} + k'}$. Note that this lower bound is maximized when $k' = 0$. Furthermore, note that $\frac{t + 1}{t - \frac{1}{t}} \ge \frac{tk + k}{tk + 1 + \frac{1}{t}}$ for any $k \ge 0$, so this lower bound is maximized when $k = 0$.
\\
Therefore, for any instance, it is sufficient for each newly-revealed vertex to carry budget $B = \frac{t + 1}{t - \frac{1}{t}} = \frac{t}{t - 1}$. Thus, $\tas_t$ is $t$-competitive while incurring at most $\frac{t}{t - 1}$ amortized recourse.
\end{proof}

\bigskip

\thmISlb*
\begin{proof}
Consider any $t$-competitive online algorithm against an adversary that constructs a complete bipartite graph and only reveals new vertices in the partition which does not contain the algorithm's current solution.
\\
We assume that the algorithm is ``sane'', in that it will not reduce the size of of its solution. Thus, whenever the algorithm switches solution and changes partition, it incurs recourse of at least twice the size of its current solution. In addition, whenever such a switch occurs, the number of vertices in the graph is at most $2t$ times the size of the algorithm's solution prior to switching. Thus, each partition-changing switch will incur at least $\frac{1}{t}$ recourse amortized over the size of the revealed graph when the switch occurs.
\\
Then, consider the largest possible increase in the size of the revealed graph between two consecutive switches. The least possible number of vertices at switch $i$ relative to $ALG$ is $(1+t)ALG$, whereas the largest possible number of vertices at switch $i+1$ is $(t^2+t)ALG$. Thus, there are at most $\frac{(t^2+t)ALG}{(1+t)ALG} = t$ times more vertices at switch $i+1$ than at switch $i$.
\\
Combining the above two results, we derive a recurrence relation that bounds below the recourse incurred by all switches up to switch $i$ amortized over the size of the revealed graph when switch $i$ occurs: $f(i) = \frac{1}{t} + \frac{1}{t} f(i-1)$. Furthermore, consider the smallest possible amortized recourse incurred by the first switch. The online algorithm must select the first revealed vertex in order to remain competitive. Then, the algorithm can wait until at most $\lceil t \rceil + 1$ vertices are revealed in the other partition before its first switch, at which point it must construct a solution with at least two vertices for an amortized recourse of $\frac{3}{\lceil t \rceil + 1}$. Alternatively, the algorithm can switch earlier and construct a solution with only one vertex in the other partition, for an amortized recourse of at least $\frac{2}{\lceil t \rceil}$. If we assume that $t \le 2$, then the first switch for any $t$-competitive online algorithm incurs at least $1$ amortized recourse. Therefore, we can set the initial value of the aforementioned recurrence relation as $f(1) = 1$.
\\
Solving this recurrence relation, we obtain that the recourse incurred by all switches up to switch $i$ amortized over the size of the revealed graph when switch $i$ occurs is bounded below by $\frac{(t - 2)(\frac{1}{t})^{i-1} + 1}{t - 1}$. This lower bound applies for any online algorithm against the described adversary when said adversary terminates its input sequence after the algorithm's $i$-th switch.
\\
Therefore, for any $1 < t \le 2$, $\varepsilon > 0$, and $t$-competitive deterministic online algorithm, there exists an instance for which the algorithm incurs at least $\frac{1}{t-1} - \varepsilon$ amortized recourse.
\end{proof}

\section{Maximum Matching Full Proof}
\thmMCMcr*
\begin{proof}
After applying the late operations on all the augmenting path with at most $2L+1$ edges, every remaining augmenting path has length at least $2L+3 = (L+2) + (L+1)$, and the ratio of the $\opt$ size to the $\LGreedy$ size $\frac{\opt(P)}{\LGreedy(P)} \leq \frac{L+2}{L+1}$ on the component $P$.

By selecting $L = \lceil \frac{1}{t-1}\rceil-1$, the $\LGreedy$ algorithm eliminates all augmenting paths that has length at most $2L+1 = 2\lceil \frac{1}{t-1}\rceil-1$. 

Afterwards, every the remaining augmenting path has length at least $2\lceil \frac{1}{t-1}\rceil+1$, and the algorithm attains a competitive ratio at most $\frac{\lceil\frac{1}{t-1}\rceil+1}{\lceil \frac{1}{t-1}\rceil}$.

\runtitle{$t$-competitiveness.} We first show that after switching every augmenting path with length at most $2L+1$, the $\LGreedy$ algorithm attains a competitive ratio of $t$. 
Consider the disjunctive union of the algorithm's matching $M$ and the optimal matching $M^*$, $M \triangle M^*$. The disjunctive union consists of connected components. Let $M_i$ and $M_i^*$ be the edges in the $i$-th component from $M$ and from $M^*$, respectively. According to our algorithm, there are two properties of the connected components:
\begin{enumerate}[label={(\bfseries P\arabic*)}]
\setlength\itemsep{0em}
    \item In any component with odd size, $\frac{|M_i^*|}{|M_i|} \leq \frac{\lceil \frac{1}{t-1}\rceil+1}{\lceil \frac{1}{t-1}\rceil}$.
    \item In any component with even size, $|M_i^*| = |M_i|$.
\end{enumerate}
Therefore, the competitive ratio by the $\LGreedy$ algorithm is at most $\frac{\sum_i |M_i^*|}{\sum_i |M_i|} \leq \max_i \frac{|M_i^*|}{|M_i|} \leq \frac{\lceil\frac{1}{t-1}\rceil+1}{\lceil\frac{1}{t-1}\rceil} = 1 + \frac{1}{\lceil\frac{1}{t-1}\rceil} \leq 1 + \frac{1}{1/(t-1)} = 1+(t-1) = t$.
\end{proof}

\bigskip

\thmMCMub*
\begin{proof}
Consider the connected component generated by the union of edges chosen by $\LGreedy$ or by $\opt$. Let $C_i$ be the components in the graph and $\tr_i$ be the total recourse incurred by the elements in $C_i$ (from very beginning till the end), the amortized recourse given by the whole graph will be upper-bounded by $\max_i \frac{\tr_i}{|C_i|}$ (Observation~\ref{obs:sum-max}). 

Now, we analyze the upper bound of the amortized recourse for any connected component. 
First we observe that the path eliminations only happen at odd-size components with length $3, 5, 7, \cdots, 2(\lceil \frac{1}{t-1}\rceil-1)+1$ (note that we only consider the case when $1< t < 2$, thus $\lceil \frac{1}{t-1}\rceil \geq 2$). 
Moreover, in our algorithm, we first check the competitiveness before any actual movement.
Hence, for an augmenting path with $2n+1$ edges that triggers a path elimination, the recourse incurred by this path elimination is $2n$.
Therefore, for such a $2n+1$-edge augmenting path, the total recourse incurred by the $2n+1$ elements in the path is at most $1+\sum_{i=1}^n 2i = 1 + n\cdot(n+1)$ (the $1$ is from the first edge in this path, which may be late accepted). Hence, the amount of amortized recourse incurred by this component is at most $\frac{1+n\cdot(n+1)}{2n+1}$.
Since the amortized recourse incurred by an augmenting path increases as the augmenting path gets longer and the $\LGreedy$ algorithm eliminates augmenting paths with length at most $2(\lceil\frac{1}{t-1}\rceil-1)+1$, the amortized recourse for a component is at most $\frac{(\lceil\frac{1}{t-1}\rceil-1)\cdot(\lceil\frac{1}{t-1}\rceil-1+1)+1}{2(\lceil\frac{1}{t-1}\rceil-1)+1} = \frac{(\lceil\frac{1}{t-1}\rceil-1)\cdot\lceil\frac{1}{t-1}\rceil+1}{2\lceil\frac{1}{t-1}\rceil-1}$.

There are two cases of the target competitive ratio $t$, $t = 1 + \frac{1}{j}$ for some integer $j$ or otherwise. 
First, we consider the case when $t = 1+\frac{1}{j}$. 
In this case, the amortized recourse for a component is at most 
$\frac{(\lceil\frac{1}{t-1}\rceil-1)\cdot\lceil\frac{1}{t-1}\rceil+1}{2\lceil\frac{1}{t-1}\rceil-1} = \frac{(\frac{1}{t-1}-1)\cdot\frac{1}{t-1}+1}{2\cdot\frac{1}{t-1}-1} = \frac{(1-(t-1))\cdot\frac{1}{t-1}+(t-1)}{2-(t-1)} = \frac{(2-t)\cdot\frac{1}{t-1}+(t-1)}{3-t} = \frac{(2-t)+(t-1)^2}{(t-1)(3-t)} = \frac{2-t}{(t-1)(3-t)} + \frac{t-1}{3-t}$.

For the case when there is no integer $j$ such that $t = 1 + \frac{1}{j}$, we round down $t$ to the largest $t^*\leq t$ such that $t^* = 1+\frac{1}{j}$ for some integer $j$. 
By eliminating all augmenting paths that have length at most $\frac{2}{t^*-1}-1 = \frac{2}{1+\frac{1}{j}-1}-1 = 2j-1$, the amount of incurred amortized recourse is at most $\frac{2-t^*}{(t^*-1)(3-t^*)} + \frac{t^*-1}{3-t^*}$, and the algorithm attains a competitive ratio of $t^* \leq t$.
\end{proof}

\bigskip

\thmMCMlb*
\begin{proof}
If $\frac{n+2}{n+1} \leq t < \frac{n+1}{n}$ for some integer $n \geq 1$, we release a sequence of $2n+1$ edges that form a path.

\runtitle{Invariants.} Consider any $1 \leq k \leq n$, the following invariants hold for any $t$-competitive algorithm:
\begin{enumerate}[label={(\bfseries I\arabic*)}]
\setlength\itemsep{0em}
    \item For a path with length $2k+1$, a $t$-competitive algorithm has to accept $k+1$ edges.
    \item For a path with length $2k$, a $t$-competitive algorithm has to accept $k$ edges.
    \item When an instance is increased from a $2(k-1)+1$ path to a $2k+1$ path, a $t$-competitive algorithm incurs at least $2k$ amount of recourse.
\end{enumerate}
Given the invariants I1, I2, and I3, the $2n+1$-path instance incurs recourse with total amount at least $\sum_{k=1}^n (2k) = n\cdot(n+1)$. Therefore, any $t$-competitive algorithm incurs at least $\frac{n\cdot(n+1)}{2n+1}$ amortized recourse for this $2n+1$-path instance. Let $t^* = \frac{n+2}{n+1} \leq t$. It follows that $n = \frac{2-t^*}{t^*-1}$. 
The amortized recourse is at least $\frac{n\cdot(n+1)}{2n+1} = \frac{\frac{2-t^*}{t^*-1}\cdot(\frac{2-t^*}{t^*-1}+1)}{2\cdot\frac{2-t^*}{t^*-1}+1} = \frac{1}{t^*-1}\cdot\frac{(2-t^*)((2-t^*)+(t^*-1))}{2\cdot(2-t^*)+(t^*-1)} = \frac{1}{t^*-1}\cdot\frac{2-t^*}{3-t^*} = \frac{(2-t^*)}{(t^*-1)(3-t^*)}$. 
Therefore, the total amortized recourse is at least $\frac{n\cdot(n+1)}{2n+1} = \frac{(2-t^*)}{(t^*-1)(3-t^*)}$.

\medskip

\runtitle{Proof of the invariants.}
Now we prove that the invariants are true throughout the instance. Recall that $\frac{n+2}{n+1} \leq t < \frac{n+1}{n}$ and $1\leq k \leq n$.
\begin{itemize}
\setlength\itemsep{0em}
    \item Proof of (I1): For any path with $2k+1$ edges, an $t$-competitive algorithm should accept at least $\frac{\opt}{t}$ edges, where $\opt$ is the cost of the offline optimal solution and has $k+1$ edges. Hence, the online algorithm should accept $\geq\frac{\opt}{t} = \frac{k+1}{t} > \frac{k+1}{(n+1)/n} = \frac{n\cdot(k+1)}{n+1} \geq \frac{k\cdot(k+1)}{k+1} = k$ edges (the last inequality is because $k\leq n$). There is a strictly greater and hence the online algorithm should accept $k+1$ edges for this $2k+1$-path. Since the optimal solution has $k+1$ edges, this means that the algorithm must accept exactly $k+1$ edges.
    \item Proof of (I2): For any path with $2k$ edges, the optimal solution has $k$ edges. Therefore, a $t$-competitive algorithm should accept at least $\frac{\opt}{t} = \frac{k}{t}$ edges. Since $t < \frac{n+1}{n}$, then $\frac{k}{t} > \frac{k}{(n+1)/n} = \frac{n\cdot k}{n+1} \geq \frac{k\cdot k}{k+1} = \frac{(k+1)(k-1)+1}{k+1} = (k-1) + \frac{1}{k+1}$. 
    Since $k > 0$, then $0<\frac{1}{k+1}<1$. Therefore, the algorithm should accept at least $(k-1)+\lceil\frac{1}{k+1}\rceil = k$ edges. Since the optimal solution has $k$ edges, this means that the algorithm must accept exactly $k$ edges.
    \item The invariant (I3) can be proven by a two-step growing for a $2k-1$ path to a $2k+1$ one. First, for the $2k-1$ path, from (I1), the algorithm accepts $k$ edges. When one more edge arrives and the instance is a $2k$-path, the algorithm has to accept exactly $k$ edges. At this step, the algorithm might incur recourse or not. Finally, when another edge arrives and the instance is a $2k+1$-path, the algorithm has to accept $k+1$ edges, which form a disjoint set with the previously accepts ones. Hence, the algorithm incurs $2k$ recourse at this step (if the algorithm is smart enough to accept the last revealed edge instead of late accepting it).
\end{itemize}
\end{proof}

\section{Vertex Cover (Full Version)}

In this section, we propose a special version of the $\tas_t$ algorithm, $\texttt{Duo-Halve}$, that attains a competitive ratio of $2-\frac{2}{\opt}$ for the $\MVC$ problem with optimal vertex cover size $\opt$.
During the process, the algorithm maintains a maximal matching $M(\ins)$ on the current input graph $\ins$ to construct a solution $\PEA(\ins)$ (we omit the parameter $\ins$ when the context is clear). 
The algorithm only selects vertices that are saturated by the matching and rejects as many vertices as possible from the two latest matched edges. 
It is clear that if we reject $2$ such vertices, the competitive ratio of the algorithm is at most $2-\frac{2}{\opt}$ since $\opt \geq |M|$. 
We show that if we cannot reject $2$ such vertices, the optimal solution size must be big so the competitive ratio is still $2-\frac{2}{\opt}$ (Theorem~\ref{thm:VCCompRatio_new}).

In the following discussion, we use some terminology. 
Let $\me1$ and $\me2$ be the most and the second-most recently matched edges respectively. Also, let $\vmi$ be the vertices saturated by the maximal matching $M(\ins)$.
The $\pea$ algorithm partitions the vertices into three groups: \textbf{Group-1}: the endpoints of $\me1$ or $\me2$, \textbf{Group-2}: the vertices in $\vm$ but not in Group-1, and \textbf{Group-3}: the vertices in $V\setminus\vm$.

A matched edge is \emph{full} if both of its endpoints are selected by the algorithm $\pea$. Otherwise, the edge is \emph{half}. 
The algorithm \emph{halves} a matched edge $(u,v)$ by producing a valid vertex cover while only accepting either $u$ or $v$.
A matched edge $(u,v)$ is \emph{halvable} if there exists a vertex cover that contains exactly one of $u$ and $v$. 
A half edge \emph{flips} if the accept/reject status of its endpoints is swapped.
A \emph{configuration} of a set of edges is a set of accept/reject statuses associated with each endpoint of those edges.

\runtitle{\texttt{Duo-Halve} Algorithm ($\PEA$).} 
When a new vertex $v$ arrives, if an edge $(p, v)$ is added to $M(\ins)$, then it introduces a new $\me1$ (namely $(p, v)$). The algorithm first accepts all \textbf{Group-2} vertices that are adjacent to $v$. 
Then, the algorithm decides the assignment of $\me1$ and $\me2$ by testing if they can be both halved by the \texttt{HalveBoth procedure} as follows.

If $\me1$ is half or one of its endpoints is $v$, the $\PEA$ algorithm halves $\me2$ if it is valid giving the current configuration of $\me1$ or $v$ being accepted. 
Otherwise, $\PEA$ halves $\me2$ if it is valid by flipping $\me1$ or late-accepting $p$. 
Otherwise, in the vertex cover returned by $\PEA$, $\me1$ is full and $\me2$ is unchanged. (See Figure~\ref{fig:flow} in Appendix for a detailed flow diagram about this \texttt{HalveBoth} procedure.)   




We start our analysis by showing that the $\PEA$ algorithm is valid, and that it runs in polynomial time.
Intuitively, the algorithm maintains a valid solution as it greedily covers edges using vertices in the maximal matching, with the exception of $\me1$ and $\me2$, where it carefully ensures that a feasible configuration is chosen.
Furthermore, the most computationally-expensive component of the $\PEA$ algorithm, which checks the validity of a constant number of configurations by looking at the neighborhoods of $\me1$ and $\me2$, runs in polynomial time.

\begin{restatable}{lemma}{lmVCtime}
\label{lm:vc-time}
The $\PEA$ algorithm always returns a valid vertex cover in $O(n^3)$ time, where $n$ is the number of vertices in the graph.
\end{restatable}
The $\PEA$ algorithm always returns a valid vertex cover in $O(n^3)$ time, where $n$ is the number of vertices in the graph.
\begin{proof}
We prove the validity of the $\pea$ algorithm by induction on the size of the input graph. 
Initially, when the graph is empty, the solution returned by $\pea$ is empty and feasible. 
Suppose that the vertex cover maintained by $\pea$ is valid, we show that after the arrival of a new vertex $v$, the new assignment by $\pea$ is also valid.

Let $\vm$ and $V$ denote the vertices saturated by the maximum matching maintained by $\pea$ and the set of vertices that have been released before the arrival of $v$. 
There are three classes of the new edges $(u,v)$ incident to the new vertex $v$: 1) $u \in \vm$ but not in $\me1$ or $\me2$, 2) $u \in \me1$ or in $\me2$, or 3) $u \in V\setminus \vm$.

By the definition of the $\pea$ algorithm, all class-$1$ edges $(u,v)$ will be covered by $u \in \vm$. By the definition of the \texttt{HalveBoth} procedure in the $\PEA$ algorithm, all new class-$2$ edges $(u, v)$ will be covered by $v$ (if there is a new matched edge introduced by the arrival of $v$) or by $u$. Class-$3$ edges only occur when there is a new matched edge $(p,v)$, where $p$ may or may not be $u$. In this case, the \texttt{HalveBoth} procedure makes sure that the edge is covered by~$v$ (or by~$u$ if the edge $(u,v)$ is added into the matching). Therefore, the $\pea$ algorithm returns a valid vertex cover.

Whenever a new vertex $v$ arrives, the $\pea$ algorithm checks all its adjacent vertices $u$, accepts them if $(u,v)$ is in class $1$, and then runs procedure \texttt{HalveBoth}. The \texttt{HalveBoth} procedure checks the coverage of edges where at least one of their endpoints is in $\me1$ or $\me2$ for a constant number of possible configurations of $\me1$ and $\me2$ being full or half. 
Hence, the total time complexity incurred by the arrival of a single vertex is $O(|E|)$ where $|E|$ is the number of edges in the final graph. Therefore, the total time complexity is $O(n^3)$.
\end{proof}

After establishing the correctness of the $\PEA$ algorithm, we describe a simple yet crucial observation about the relationship between $|M|$, $\opt$, and $\frac{\PEA}{\opt}$.

\begin{observation}\label{obs:subset}
If $\opt \geq |M|+1$, then $\frac{\PEA}{\opt} \leq 2-\frac{2}{\opt}$.
\end{observation}
\hide{\begin{proof}
Since the $\PEA$ algorithm only accepts the vertices saturated by $M$, $\PEA \le 2|M|$. Therefore, $\frac{\PEA}{\opt} \le \frac{2|M|}{|M| + 1} = \frac{2|M|+2-2}{|M|+1}$ $= 2-\frac{2}{|M|+1}\leq 2-\frac{2}{\opt}$.
\end{proof}}

The above observation allows us to show another condition under which $\opt \ge |M| + 1$ and thus our desired bound of $2 - \frac{2}{\opt}$ on the competitive ratio is achieved.
The contrapositive of this lemma also provides us with a necessary characteristic of any ``problematic'' instance.

\begin{restatable}{lemma}{lmVCsubset}
\label{lm:subset}
If $\opt \setminus \vm \neq \varnothing$ or there exists an edge $(u,v)$ in $M$ such that $\{u,v\}\subseteq \opt$, then $\frac{\PEA}{\opt} \leq 2-\frac{2}{\opt}$.
\end{restatable}
\begin{proof}
Assume there exists at least one vertex in $\opt$ but not in $\alg$. Since $\opt$ must select one vertex from each matched edges in $M$, $\opt \geq |M|+1$.

If there is a matched edge whose endpoints are both selected by $\opt$, since $\opt$ must select one vertex from each of the other matched edges, $\opt \geq 2 + (|M|-1) = |M| + 1$.

By Observation~\ref{obs:subset}, $\frac{\PEA}{\opt} \leq 2-\frac{2}{\opt}$ in both of the cases.
\end{proof}

\begin{corollary}\label{col:subset}
If $2-\frac{2}{\opt} < \frac{\PEA}{\opt}$, then $\opt \subseteq \vm$.
\end{corollary}
\hide{
\begin{proof}
This corollary follows directly as the contrapositive to Lemma~\ref{lm:subset}.
\end{proof}
}

We then present another observation, which helps with our analysis of $\PEA$ by providing a feasibility guarantee for any newly-revealed $\me1$.

\begin{observation}\label{obs:shift-me1}
Immediately after an edge is added to the matching, $\me1$ is halvable by accepting the newly-revealed vertex.
\end{observation}
\hide{\begin{proof}
Let $\me1 = (p, v)$, where $v$ is the newly-revealed vertex.
Prior to $v$ being revealed, $p$ is unmatched and thus rejected.
Furthermore, accepting $v$ covers all newly-revealed edges.
Therefore, $\me1$ is halvable by accepting $v$ and rejecting $p$.
\end{proof}}

Prior to bounding the competitiveness of the $\PEA$ algorithm, we show one more structural property of this algorithm.
Namely, if $\PEA$ fails to produce a solution in which it accepts only one vertex of $\me1$, then there is no possible solution contained in $\vm$ that only accepts one vertex of $\me1$.
In particular, this applies to $\opt$ as well, which by Lemma~\ref{lm:subset} means that our desired bound on the competitive ratio is achieved.

We achieve this by considering the rejected vertices in the neighborhoods of the endpoints of $\me1$.
Specifically, we show that if $\me1$ is full, then each of these neighborhoods must contain a vertex outside of $\vm$, and thus $\opt$ must either contain a vertex outside of $\vm$ or both endpoints of $\me1$.

\lmVCfull*
\begin{proof}
Assume that $\me1 = (p, v)$ (where $v$ was revealed later) is full and consider the ``blocking set'' of each of its endpoints, i.e. the set of rejected vertices in its neighborhood.
We partition the possible blocking set vertices into three classes: 1) vertices in $\vm$ but not in $\me2$, 2) vertices in $\me2$, and 3) vertices in $V\setminus\vm$. We prove the lemma by showing the following claim.

\runtitle{Claim:} If $\me1$ is full, then both the blocking set of $p$ and of $v$ contain at least one vertex in $V\setminus \vm$.

Let $x$ and $y$ be a class-$3$ vertex in the blocking set of $p$ and of $v$, respectively (it is possible that $x=y$).
If the claim is true, then the optimal solution has to choose at least two vertices in $\{p,v,x,y\}$, since the three edges $(p,v)$, $(p,x)$, and $(y,v)$ must be covered. This proves the lemma.

Now, we prove the correctness of the claim by showing that the blocking sets of $p$ and of $v$ cannot contain class-$1$ vertices, and cannot contain only a class-$2$ vertex.
Note that a blocking set cannot contain more than one class-$2$ vertex, as that would imply that both endpoints of $\me2$ are rejected.,

First, consider any class-$1$ blocking set vertex $x$, which is in $\vm$ but not in $\me2$.
Immediately prior to $(p, v)$ being added to the matching, all edges $(p,x)$ are covered by $x$ since $x$ must be accepted.
When $v$ is revealed, all of the considered vertices $x$ are (late-)accepted, and all edges $(v,x)$ are covered by $x$.
Furthermore, none of the considered vertices will be late-rejected after $(p, v)$ is added to the matching (by definition of $\PEA$).
Thus, all edges between $\{p,v\}$ and the considered vertices are covered. Hence, none of the considered vertices can be in the blocking set of either $p$ or $v$.

Next, assume that either blocking set contains only a vertex in $\me2$.
By the definition of the \texttt{HalveBoth} function in the $\PEA$ algorithm, if a valid configuration exists such that $\me1$ is half, then it would have been chosen by $\PEA$.
However, if either $p$ or $v$ is blocked only by a vertex in $\me2$, $\me1$ is halvable by making $\me2$ full.
This is a contradiction, so each blocking set must contain at least one vertex outside of $\vm$, and the claim is proven.
\end{proof}

\begin{corollary}\label{col:full-me1}
In the assignment of $\pea$, if $\me1$ is full, then $ \frac{\pea}{\opt} \leq 2-\frac{2}{\opt}$.
\end{corollary}
\hide{
\begin{proof}
By Lemma~\ref{lm:full-me1} and Observation~\ref{obs:subset}.
\end{proof}
}

After taking care of the above preliminaries, we show that our desired bound of $2 - \frac{2}{\opt}$ on the competitiveness of the $\PEA$ algorithm is indeed attainable.

\VCcr*
\begin{proof}
We prove the theorem by considering three possible states for $\PEA$'s solution: 1) both $\me1$ and $\me2$ are half, 2) $\me1$ is half and $\me2$ is full, and 3) $\me1$ is full.
In each case, we show that either $\PEA \le 2|M| - 2 \le 2\opt - 2$ or $\opt \ge |M| + 1$, both of which imply that $\frac{\PEA}{\opt} \le 2 - \frac{2}{\opt}$.
In state $1$, there are at least two half edges in $M$, so $\PEA \le 2|M| - 2$ and thus $\frac{\PEA}{\opt} \le 2 - \frac{2}{\opt}$.
In state $3$, by Corollary~\ref{col:full-me1}, $\frac{\PEA}{\opt} \le 2 - \frac{2}{\opt}$.
Therefore, we consider state $2$ for the remainder of the proof.

Then, assume that $\me2$ is full and consider the ``blocking set'' of each of its endpoints, i.e. the set of rejected vertices in its neighborhood. 
We partition the possible blocking set vertices into three classes: 1) vertices in $\vm$ but not in $\me1$, 2) vertices in $\me1$, and 3) vertices in $V\setminus\vm$.
We use this partition to show that, based on the composition of each blocking set, either there are two half edges and $\PEA \le 2|M| - 2 \le 2\opt - 2$, or $\opt \ge |M| + 1$ (which implies that $\frac{\PEA}{\opt} \le 2 - \frac{2}{\opt}$ by Observation~\ref{obs:subset}).
This proves the theorem.

First, consider the vertices in $\vm$ that are not in $\me1$.
If any of these vertices is in the blocking set of either endpoint of $\me2$, then there must be at least one half edge in $M$ other than $\me1$, which implies that $\pea \leq 2|M|-2 \leq 2\opt-2$. 

Then, assume that either blocking set contains only a vertex in $\me1$.
By the definition of the \texttt{HalveBoth} function in the $\PEA$ algorithm, if a valid configuration exists such that both $\me1$ and $\me2$ are half (with all other edges in $M$ full as per the previous paragraph), then it would have been chosen by $\PEA$.
Therefore, any vertex cover must contain either both vertices in $\me1$ or both vertices in $\me2$.
Thus, $\opt \ge |M| + 1$, which implies $\frac{\PEA}{\opt} \le 2 - \frac{2}{\opt}$ by Observation~\ref{obs:subset}.

Finally, if the blocking set for each vertex in $\me2$ contains a vertex outside of $\vm$, $\opt$ must contain either a vertex outside of $\vm$ or both vertices in $\me2$.
This implies that $\opt \ge |M| + 1$ and thus $\frac{\PEA}{\opt} \le 2 - \frac{2}{\opt}$ by Observation~\ref{obs:subset}.

Therefore there is no instance such that $\frac{\PEA}{\opt} > 2-\frac{2}{\opt}$, and thus the $\PEA$ algorithm is $(2-\frac{2}{\opt})$-competitive.
\end{proof}

With the upper bound on the competitive ratio of the $\PEA$ algorithm proven, we then show an upper bound on the amortized recourse incurred by $\PEA$.
Prior to doing so, we show a restriction on the possible changes to the $\PEA$ algorithm's solution when a new edge is added to the maximal matching $M$.
By extension, this limits the possible number of late operations incurred in this scenario.

\begin{restatable}{lemma}{lmVCshift}
\label{lm:shift-me2}
When an edge shifts from $\me1$ to $\me2$, it either goes from full to half or remains unchanged.
\end{restatable}
\begin{proof}
First, note that when a shift occurs, it is because a new edge $\me1$ is added to the matching, and this edge is incident to the newly-revealed vertex $v$.

Let $C$ denote the vertices accepted by $\pea$ prior to the arrival of $v$.
The $\pea$ algorithm maintains a valid vertex cover throughout the process (by Lemma~\ref{lm:vc-time}), so by Observation~\ref{obs:shift-me1}, $C\cup \{v\}$ is a valid vertex cover.
Therefore, the $\pea$ algorithm accepts $v$ and does not change the assignment of $\me2$ if $\me2$ was half (by the procedure \texttt{HalveBoth}).
That is, the accept/reject status of endpoints of $\me2$ is not swapped, and $\me2$ cannot go from half to full.

Therefore, when an edge shifts from $\me1$ to $\me2$, it either goes from full to half or remains unchanged.
\end{proof}

In particular, Lemma~\ref{lm:shift-me2} tells us that a half edge shifting from $\me1$ to $\me2$ will not incur any late operations (i.e. it won't become full and it won't flip).

In the worst case, there can be multiple late accepts on \textbf{Group-2} vertices
, and $4$ late operations on vertices in $\me1$ and $\me2$. In the following theorem, we use a potential function to prove that the amortized late operations per vertex is at most $\frac{10}{3}$.

\VCar*
\begin{proof}
We prove the theorem by using a potential function. 
For each vertex $v_i$, let $\lo_i$ denote the number of incurred late operations when it arrives. Assume there exists a potential function $\Phi$ such that the potential after the arrival of $v_i$ and all the incurred late operations is $\Phi_i \geq \Phi_0$, and $\lo_i + \Phi_i-\Phi_{i-1} \leq c$. Then, the total number of late operations is upper-bounded by $c\cdot n + \Phi_0 - \Phi_n \leq c\cdot n$, where $n$ is the number of vertices. Thus, the amortized recourse is bounded above by $c$.

At any given point in the input sequence where the matching constructed by $\PEA$ contains at least $2$ edges, the status of $\me1 = (p, v)$ and $\me2 = (u, w)$ is characterized by one of the following $6$ states, where $v$ is the vertex of $\me1$ revealed last:
\begin{enumerate}
    \item $u$, $w$, and $p$ are accepted
    \item $u$, $w$, and $v$ are accepted
    \item $u$, $w$, $p$, and $v$ are accepted
    \item $p$ and one of $u$ and $w$ are accepted
    \item $v$ and one of $u$ and $w$ are accepted
    \item $p$, $v$, and one of $u$ and $w$ are accepted
\end{enumerate}

Furthermore, define a half edge $(u, v)$ as being \emph{free} if it is halvable both by accepting $u$ and by accepting $v$.
Also, define a half edge as being \emph{expired} if it is neither $\me1$ nor $\me2$.
Finally, define $A$ as the set of vertices accepted by $\PEA$.
Then, define the potential function $\Phi$ as
\begin{equation*}
\Phi := |\{(u, v) | (u, v) \mbox{ expired half}\}| + \frac{1}{3}|A \cap (\me1 \cup \me2)| + \frac{2}{3} \cdot \mathds{1}[\mbox{free half } \me2]
\end{equation*}

Now, we show that, for any possible state transition triggered by a newly-revealed vertex, the number of incurred late operations $\lo$ added to the change in potential $\dphi$ is bounded above by $\frac{10}{3}$.
Note that, for any newly-revealed vertex $v$, $v$ may be adjacent to $k \ge 0$ rejected vertices that are matched by some expired edge.
This incurs $k$ late operations, but also decreases $\Phi$ by $k$, so this may be ignored when computing $\lo + \dphi$.

\begin{itemize}
    \item $1 \to 1$, without shift: $\lo + \dphi \le 0 + (0 + 0 + 0) = 0$
    \item $1 \to 1$, with shift: not possible, because an edge shifting from $\me1$ to $\me2$ won't go from half to full (Lemma~\ref{lm:shift-me2}) 
    \item $1 \to 2$, without shift: $\lo + \dphi \le 2 + (0 + 0 + 0) = 2$
    \item $1 \to 2$, with shift: not possible, because an edge shifting from $\me1$ to $\me2$ won't go from half to full (Lemma~\ref{lm:shift-me2}) 
    \item $1 \to 3$, without shift: $\lo + \dphi \le 1 + (0 + \frac{1}{3} + 0) = \frac{4}{3}$
    \item $1 \to 3$, with shift: not possible, because $\me1$ is always halvable after a shift (Observation~\ref{obs:shift-me1}) 
    \item $1 \to 4$, without shift: $\lo + \dphi \le 1 + (0 - \frac{1}{3} + \frac{2}{3}) = \frac{4}{3}$
    \item $1 \to 4$, with shift: $\lo + \dphi \le 1 + (0 - \frac{1}{3} + \frac{2}{3}) = \frac{4}{3}$
    \item $1 \to 5$, without shift: $\lo + \dphi \le 3 + (0 - \frac{1}{3} + \frac{2}{3}) = \frac{10}{3}$
    \item $1 \to 5$, with shift: $\lo + \dphi \le 0 + (0 - \frac{1}{3} + \frac{2}{3}) = \frac{1}{3}$
    \item $1 \to 6$, without shift: not possible, because $\PEA$ won't halve $\me2$ if $\me1$ is full 
    \item $1 \to 6$, with shift: not possible, because $\PEA$ prioritizes halving $\me1$ and this is always possible after a shift (Observation~\ref{obs:shift-me1}) 
    
    \item $2 \to 1$, without shift: $\lo + \dphi \le 2 + (0 + 0 + 0) = 2$
    \item $2 \to 1$, with shift: not possible, because an edge shifting from $\me1$ to $\me2$ won't go from half to full (Lemma~\ref{lm:shift-me2}) 
    \item $2 \to 2$, without shift: $\lo + \dphi \le 0 + (0 + 0 + 0) = 0$
    \item $2 \to 2$, with shift: not possible, because an edge shifting from $\me1$ to $\me2$ won't go from half to full (Lemma~\ref{lm:shift-me2}) 
    \item $2 \to 3$, without shift: $\lo + \dphi \le 1 + (0 + \frac{1}{3} + 0) = \frac{4}{3}$
    \item $2 \to 3$, with shift: not possible, because $\me1$ is always halvable after a shift (Observation~\ref{obs:shift-me1}) 
    \item $2 \to 4$, without shift: $\lo + \dphi \le 3 + (0 - \frac{1}{3} + \frac{2}{3}) = \frac{10}{3}$
    \item $2 \to 4$, with shift: $\lo + \dphi \le 1 + (0 - \frac{1}{3} + \frac{2}{3}) = \frac{4}{3}$
    \item $2 \to 5$, without shift: $\lo + \dphi \le 1 + (0 - \frac{1}{3} + \frac{2}{3}) = \frac{4}{3}$
    \item $2 \to 5$, with shift: $\lo + \dphi \le 0 + (0 - \frac{1}{3} + \frac{2}{3}) = \frac{1}{3}$
    \item $2 \to 6$, without shift: not possible, because $\PEA$ won't halve $\me2$ if $\me1$ is full 
    \item $2 \to 6$, with shift: not possible, because $\PEA$ prioritizes halving $\me1$ and this is always possible after a shift (Observation~\ref{obs:shift-me1}) 
    
    \item $3 \to 1$, without shift: $\lo + \dphi \le 1 + (0 - \frac{1}{3} + 0) = \frac{2}{3}$
    \item $3 \to 1$, with shift: $\lo + \dphi \le 1 + (0 - \frac{1}{3} + 0) = \frac{2}{3}$
    \item $3 \to 2$, without shift: $\lo + \dphi \le 1 + (0 - \frac{1}{3} + 0) = \frac{2}{3}$
    \item $3 \to 2$, with shift: $\lo + \dphi \le 0 + (0 - \frac{1}{3} + 0) = -\frac{1}{3}$
    \item $3 \to 3$, without shift: $\lo + \dphi \le 0 + (0 + 0 + 0) = 0$
    \item $3 \to 3$, with shift: not possible, because $\me1$ is always halvable after a shift (Observation~\ref{obs:shift-me1}) 
    \item $3 \to 4$, without shift: $\lo + \dphi \le 2 + (0 - \frac{2}{3} + \frac{2}{3}) = 2$
    \item $3 \to 4$, with shift: $\lo + \dphi \le 2 + (0 - \frac{2}{3} + \frac{2}{3}) = 2$
    \item $3 \to 5$, without shift: $\lo + \dphi \le 2 + (0 - \frac{2}{3} + \frac{2}{3}) = 2$
    \item $3 \to 5$, with shift: $\lo + \dphi \le 1 + (0 - \frac{2}{3} + \frac{2}{3}) = 1$
    \item $3 \to 6$, without shift: not possible, because $\PEA$ won't halve $\me2$ if $\me1$ is full 
    \item $3 \to 6$, with shift: not possible, because $\PEA$ prioritizes halving $\me1$ and this is always possible after a shift (Observation~\ref{obs:shift-me1}) 
    
    \item $4 \to 1$, without shift: $\lo + \dphi \le 1 + (0 + \frac{1}{3} + 0) = \frac{4}{3}$
    \item $4 \to 1$, with shift: not possible, because an edge shifting from $\me1$ to $\me2$ won't go from half to full (Lemma~\ref{lm:shift-me2}) 
    \item $4 \to 2$, without shift: $\lo + \dphi \le 3 + (0 + \frac{1}{3} + 0) = \frac{10}{3}$
    \item $4 \to 2$, with shift: not possible, because an edge shifting from $\me1$ to $\me2$ won't go from half to full (Lemma~\ref{lm:shift-me2}) 
    \item $4 \to 3$, without shift: $\lo + \dphi \le 2 + (0 + \frac{2}{3} + 0) = \frac{8}{3}$
    \item $4 \to 3$, with shift: not possible, because $\me1$ is always halvable after a shift (Observation~\ref{obs:shift-me1}) 
    \item $4 \to 4$, without shift: $\lo + \dphi \le 2 + (0 + 0 - \frac{2}{3}) = \frac{4}{3}$ 
    \item $4 \to 4$, with shift: $\lo + \dphi \le 1 + (1 + 0 + \frac{2}{3}) = \frac{8}{3}$
    \item $4 \to 5$, without shift: $\lo + \dphi \le 4 + (0 + 0 - \frac{2}{3}) = \frac{10}{3}$ 
    \item $4 \to 5$, with shift: $\lo + \dphi \le 0 + (1 + 0 + \frac{2}{3}) = \frac{5}{3}$
    \item $4 \to 6$, without shift: $\lo + \dphi \le 3 + (0 + \frac{1}{3} - \frac{2}{3}) = \frac{8}{3}$ 
    \item $4 \to 6$, with shift: not possible, because $\PEA$ prioritizes halving $\me1$ and this is always possible after a shift (Observation~\ref{obs:shift-me1}) 
    
    \item $5 \to 1$, without shift: $\lo + \dphi \le 3 + (0 + \frac{1}{3} + 0) = \frac{10}{3}$
    \item $5 \to 1$, with shift: not possible, because an edge shifting from $\me1$ to $\me2$ won't go from half to full (Lemma~\ref{lm:shift-me2}) 
    \item $5 \to 2$, without shift: $\lo + \dphi \le 1 + (0 + \frac{1}{3} + 0) = \frac{4}{3}$
    \item $5 \to 2$, with shift: not possible, because an edge shifting from $\me1$ to $\me2$ won't go from half to full (Lemma~\ref{lm:shift-me2}) 
    \item $5 \to 3$, without shift: $\lo + \dphi \le 2 + (0 + \frac{2}{3} + 0) = \frac{8}{3}$
    \item $5 \to 3$, with shift: not possible, because $\me1$ is always halvable after a shift (Observation~\ref{obs:shift-me1}) 
    \item $5 \to 4$, without shift: $\lo + \dphi \le 4 + (0 + 0 - \frac{2}{3}) = \frac{10}{3}$ 
    \item $5 \to 4$, with shift: $\lo + \dphi \le 1 + (1 + 0 + \frac{2}{3}) = \frac{8}{3}$
    \item $5 \to 5$, without shift: $\lo + \dphi \le 2 + (0 + 0 - \frac{2}{3}) = \frac{4}{3}$ 
    \item $5 \to 5$, with shift: $\lo + \dphi \le 0 + (1 + 0 + \frac{2}{3}) = \frac{5}{3}$
    \item $5 \to 6$, without shift: $\lo + \dphi \le 3 + (0 + \frac{1}{3} - \frac{2}{3}) = \frac{8}{3}$ 
    \item $5 \to 6$, with shift: not possible, because $\PEA$ prioritizes halving $\me1$ and this is always possible after a shift (Observation~\ref{obs:shift-me1}) 
    
    \item $6 \to 1$, without shift: $\lo + \dphi \le 2 + (0 + 0 + 0) = 2$
    \item $6 \to 1$, with shift: $\lo + \dphi \le 1 + (1 + 0 + 0) = 2$
    \item $6 \to 2$, without shift: $\lo + \dphi \le 2 + (0 + 0 + 0) = 2$
    \item $6 \to 2$, with shift: $\lo + \dphi \le 0 + (1 + 0 + 0) = 1$
    \item $6 \to 3$, without shift: $\lo + \dphi \le 1 + (0 + \frac{1}{3} + 0) = \frac{4}{3}$
    \item $6 \to 3$, with shift: not possible, because $\me1$ is always halvable after a shift (Observation~\ref{obs:shift-me1}) 
    \item $6 \to 4$, without shift: $\lo + \dphi \le 3 + (0 - \frac{1}{3} - \frac{2}{3}) = 2$ 
    \item $6 \to 4$, with shift: $\lo + \dphi \le 2 + (1 - \frac{1}{3} + \frac{2}{3}) = \frac{10}{3}$
    \item $6 \to 5$, without shift: $\lo + \dphi \le 3 + (0 - \frac{1}{3} - \frac{2}{3}) = 2$ 
    \item $6 \to 5$, with shift: $\lo + \dphi \le 1 + (1 - \frac{1}{3} + \frac{2}{3}) = \frac{7}{3}$
    \item $6 \to 6$, without shift: $\lo + \dphi \le 2 + (0 + 0 - \frac{2}{3}) = \frac{4}{3}$ 
    \item $6 \to 6$, with shift: not possible, because $\PEA$ prioritizes halving $\me1$ and this is always possible after a shift (Observation~\ref{obs:shift-me1}) 
\end{itemize}

Thus, for any possible state transition, $\lo + \dphi \le \frac{10}{3}$.
Furthermore, $\Phi_0 = 0$ and $\Phi_i \ge 0$.
Therefore, the amortized recourse incurred by $\PEA$ is bounded above by $\frac{10}{3}$.
\end{proof}

After establishing an upper bound on the amortized recourse incurred by the $\PEA$ algorithm, we show a lower bound by constructing a family of instances that alternates between incurring a late accept on a \textbf{Group-2} vertex, 
and $4$ late operations on $\me1$ and $\me2$.
This is illustrated in Figure~\ref{fig:vc-ar-lb} (see Appendix).

\VCARLB*
\hide{
\begin{proof}
Consider $\PEA$ against an instance that reveals vertices as follows:
\begin{enumerate}
    \item isolated vertex $a$. $\PEA$ rejects $a$.
    \item $b$ adjacent to $a$. $\PEA$ adds $(a, b)$ to the matching, and accepts $b$.
    \item $c$ adjacent to $a$. $\PEA$ late-accepts $a$, and late-rejects $b$.
    \item $d$ adjacent to $c$. $\PEA$ adds $(c, d)$ to the matching, and accepts $d$.
    \item $e$ adjacent to $b$ and $c$. $\PEA$ late-accepts $b$ and $c$, and late-rejects $a$ and $d$.
    \item $f$ adjacent to $e$ and $a$. $\PEA$ adds $(e, f)$ to the matching, accepts $f$, and late-accepts $a$.
\end{enumerate}
This instance is illustrated in Figure~\ref{fig:vc-ar-lb} (see appendix).
Observe that the configuration of $\{c, d, e, f\}$ after $f$ is revealed is isomorphic to the configuration of $\{a, b, c, d\}$ after $d$ is revealed. Furthermore, all edges with an endpoint outside of $\{c, d, e, f\}$ are covered by a vertex outside of $\{c, d, e, f\}$.
Thus, this instance can be extended indefinitely by repeating steps $5$ and $6$.
Since these two steps together incur $5$ late operations by revealing two vertices, the family of instances described here incurs asymptotic amortized recourse $\frac{5}{2}$.
\end{proof}
}

Finally, we show that the upper bound on the competitive ratio provided by Theorem~\ref{thm:VCCompRatio_new} is tight not only for the $\PEA$ algorithm, but for a class of online algorithms defined below.
This is done by using an adversary that constructs an arbitrary number of triangles that all share a common vertex $v$.
This common vertex is revealed last, so all edges not incident to $v$ are added to the matching and $v$ is rejected.
This is illustrated in Figure~\ref{fig:vc-cr-lb}.

\begin{definition}
An algorithm for vertex cover is \emph{incremental matching-based} if it maintains a maximal matching throughout the process in an incremental manner, and its solution only contains vertices saturated by the matching.
\end{definition}

\VROPT*
\begin{proof}
Consider any incremental matching-based algorithm against an instance where $k$ disconnected edges are revealed via their endpoints.
By nature of the incremental matching-based algorithm, each of these $k$ edges will be added to the matching, and at least one vertex from each pair will be accepted.
Then, a final vertex is revealed, adjacent to all previously-revealed vertices.
The incremental matching-based algorithm will not accept this vertex, as it is not matched, but must accept all other vertices for a vertex cover of size $2k$.
However, the optimal solution consists of the last revealed vertex, and one vertex from each pair.
Thus, against this instance, any incremental matching-based algorithm will be $\frac{2k}{k+1} = 2 - \frac{2}{\opt}$ times worse than the optimal solution.
Therefore, no incremental matching-based algorithm can achieve a competitive ratio smaller than $2 - \frac{2}{\opt}$.
\end{proof}

\hide{
\section{Vertex Cover Full Proof}
\lmVCtime*
The $\PEA$ algorithm always returns a valid vertex cover in $O(n^3)$ time, where $n$ is the number of vertices in the graph.
\begin{proof}
We prove the validity of the $\pea$ algorithm by induction on the size of the input graph. 
Initially, when the graph is empty, the solution returned by $\pea$ is empty and feasible. 
Suppose that the vertex cover maintained by $\pea$ is valid, we show that after the arrival of a new vertex $v$, the new assignment by $\pea$ is also valid.

Let $\vm$ and $V$ denote the vertices saturated by the maximum matching maintained by $\pea$ and the set of vertices that have been released before the arrival of $v$. 
There are three classes of the new edges $(u,v)$ incident to the new vertex $v$: 1) $u \in \vm$ but not in $\me1$ or $\me2$, 2) $u \in \me1$ or in $\me2$, or 3) $u \in V\setminus \vm$.

By the definition of the $\pea$ algorithm, all class-$1$ edges $(u,v)$ will be covered by $u \in \vm$. By the definition of the \texttt{HalveBoth} procedure in the $\PEA$ algorithm, all new class-$2$ edges $(u, v)$ will be covered by $v$ (if there is a new matched edge introduced by the arrival of $v$) or by $u$. Class-$3$ edges only occur when there is a new matched edge $(p,v)$, where $p$ may or may not be $u$. In this case, the \texttt{HalveBoth} procedure makes sure that the edge is covered by~$v$ (or by~$u$ if the edge $(u,v)$ is added into the matching). Therefore, the $\pea$ algorithm returns a valid vertex cover.

Whenever a new vertex $v$ arrives, the $\pea$ algorithm checks all its adjacent vertices $u$, accepts them if $(u,v)$ is in class $1$, and then runs procedure \texttt{HalveBoth}. The \texttt{HalveBoth} procedure checks the coverage of edges where at least one of their endpoints is in $\me1$ or $\me2$ for a constant number of possible configurations of $\me1$ and $\me2$ being full or half. 
Hence, the total time complexity incurred by the arrival of a single vertex is $O(|E|)$ where $|E|$ is the number of edges in the final graph. Therefore, the total time complexity is $O(n^3)$.
\end{proof}

\lmVCsubset*
\begin{proof}
Assume there exists at least one vertex in $\opt$ but not in $\alg$. Since $\opt$ must select one vertex from each matched edges in $M$, $\opt \geq |M|+1$.

If there is a matched edge whose endpoints are both selected by $\opt$, since $\opt$ must select one vertex from each of the other matched edges, $\opt \geq 2 + (|M|-1) = |M| + 1$.

By Observation~\ref{obs:subset}, $\frac{\PEA}{\opt} \leq 2-\frac{2}{\opt}$ in both of the cases.
\end{proof}

\lmVCfull*
\begin{proof}
Assume that $\me1 = (p, v)$ (where $v$ was revealed later) is full and consider the ``blocking set'' of each of its endpoints, i.e. the set of rejected vertices in its neighborhood.
We partition the possible blocking set vertices into three classes: 1) vertices in $\vm$ but not in $\me2$, 2) vertices in $\me2$, and 3) vertices in $V\setminus\vm$. We prove the lemma by showing the following claim.

\runtitle{Claim:} If $\me1$ is full, then both the blocking set of $p$ and of $v$ contain at least one vertex in $V\setminus \vm$.

Let $x$ and $y$ be a class-$3$ vertex in the blocking set of $p$ and of $v$, respectively (it is possible that $x=y$).
If the claim is true, then the optimal solution has to choose at least two vertices in $\{p,v,x,y\}$, since the three edges $(p,v)$, $(p,x)$, and $(y,v)$ must be covered. This proves the lemma.

Now, we prove the correctness of the claim by showing that the blocking sets of $p$ and of $v$ cannot contain class-$1$ vertices, and cannot contain only a class-$2$ vertex.
Note that a blocking set cannot contain more than one class-$2$ vertex, as that would imply that both endpoints of $\me2$ are rejected.,

First, consider any class-$1$ blocking set vertex $x$, which is in $\vm$ but not in $\me2$.
Immediately prior to $(p, v)$ being added to the matching, all edges $(p,x)$ are covered by $x$ since $x$ must be accepted.
When $v$ is revealed, all of the considered vertices $x$ are (late-)accepted, and all edges $(v,x)$ are covered by $x$.
Furthermore, none of the considered vertices will be late-rejected after $(p, v)$ is added to the matching (by definition of $\PEA$).
Thus, all edges between $\{p,v\}$ and the considered vertices are covered. Hence, none of the considered vertices can be in the blocking set of either $p$ or $v$.

Next, assume that either blocking set contains only a vertex in $\me2$.
By the definition of the \texttt{HalveBoth} function in the $\PEA$ algorithm, if a valid configuration exists such that $\me1$ is half, then it would have been chosen by $\PEA$.
However, if either $p$ or $v$ is blocked only by a vertex in $\me2$, $\me1$ is halvable by making $\me2$ full.
This is a contradiction, so each blocking set must contain at least one vertex outside of $\vm$, and the claim is proven.
\end{proof}

\hide{
\lmVCshift*
\begin{proof}
First, note that when a shift occurs, it is because a new edge $\me1$ is added to the matching, and this edge is incident to the newly-revealed vertex $v$.

Let $C$ denote the vertices accepted by $\pea$ prior to the arrival of $v$.
The $\pea$ algorithm maintains a valid vertex cover throughout the process (by Lemma~\ref{lm:vc-time}), so by Observation~\ref{obs:shift-me1}, $C\cup \{v\}$ is a valid vertex cover.
Therefore, the $\pea$ algorithm accepts $v$ and does not change the assignment of $\me2$ if $\me2$ was half (by the procedure \texttt{HalveBoth}).
That is, the accept/reject status of endpoints of $\me2$ is not swapped, and $\me2$ cannot go from half to full.

Therefore, when an edge shifts from $\me1$ to $\me2$, it either goes from full to half or remains unchanged.
\end{proof}
}

\hide{
\VCcr*
\begin{proof}
We prove the theorem by considering three possible states for $\PEA$'s solution: 1) both $\me1$ and $\me2$ are half, 2) $\me1$ is half and $\me2$ is full, and 3) $\me1$ is full.
In each case, we show that either $\PEA \le 2|M| - 2 \le 2\opt - 2$ or $\opt \ge |M| + 1$, both of which imply that $\frac{\PEA}{\opt} \le 2 - \frac{2}{\opt}$.
In state $1$, there are at least two half edges in $M$, so $\PEA \le 2|M| - 2$ and thus $\frac{\PEA}{\opt} \le 2 - \frac{2}{\opt}$.
In state $3$, by Corollary~\ref{col:full-me1}, $\frac{\PEA}{\opt} \le 2 - \frac{2}{\opt}$.
Therefore, we consider state $2$ for the remainder of the proof.

Then, assume that $\me2$ is full and consider the ``blocking set'' of each of its endpoints, i.e. the set of rejected vertices in its neighborhood. 
We partition the possible blocking set vertices into three classes: 1) vertices in $\vm$ but not in $\me1$, 2) vertices in $\me1$, and 3) vertices in $V\setminus\vm$.
We use this partition to show that, based on the composition of each blocking set, either there are two half edges and $\PEA \le 2|M| - 2 \le 2\opt - 2$, or $\opt \ge |M| + 1$ (which implies that $\frac{\PEA}{\opt} \le 2 - \frac{2}{\opt}$ by Observation~\ref{obs:subset}).
This proves the theorem.

First, consider the vertices in $\vm$ that are not in $\me1$.
If any of these vertices is in the blocking set of either endpoint of $\me2$, then there must be at least one half edge in $M$ other than $\me1$, which implies that $\pea \leq 2|M|-2 \leq 2\opt-2$. 

Then, assume that either blocking set contains only a vertex in $\me1$.
By the definition of the \texttt{HalveBoth} function in the $\PEA$ algorithm, if a valid configuration exists such that both $\me1$ and $\me2$ are half (with all other edges in $M$ full as per the previous paragraph), then it would have been chosen by $\PEA$.
Therefore, any vertex cover must contain either both vertices in $\me1$ or both vertices in $\me2$.
Thus, $\opt \ge |M| + 1$, which implies $\frac{\PEA}{\opt} \le 2 - \frac{2}{\opt}$ by Observation~\ref{obs:subset}.

Finally, if the blocking set for each vertex in $\me2$ contains a vertex outside of $\vm$, $\opt$ must contain either a vertex outside of $\vm$ or both vertices in $\me2$.
This implies that $\opt \ge |M| + 1$ and thus $\frac{\PEA}{\opt} \le 2 - \frac{2}{\opt}$ by Observation~\ref{obs:subset}.

Therefore there is no instance such that $\frac{\PEA}{\opt} > 2-\frac{2}{\opt}$, and thus the $\PEA$ algorithm is $(2-\frac{2}{\opt})$-competitive.
\end{proof}
}

\hide{
\VCar*
\begin{proof}
We prove the theorem by using a potential function. 
For each vertex $v_i$, let $\lo_i$ denote the number of incurred late operations when it arrives. Assume there exists a potential function $\Phi$ such that the potential after the arrival of $v_i$ and all the incurred late operations is $\Phi_i \geq \Phi_0$, and $\lo_i + \Phi_i-\Phi_{i-1} \leq c$. Then, the total number of late operations is upper-bounded by $c\cdot n + \Phi_0 - \Phi_n \leq c\cdot n$, where $n$ is the number of vertices. Thus, the amortized recourse is bounded above by $c$.

At any given point in the input sequence where the matching constructed by $\PEA$ contains at least $2$ edges, the status of $\me1 = (p, v)$ and $\me2 = (u, w)$ is characterized by one of the following $6$ states, where $v$ is the vertex of $\me1$ revealed last:
\begin{enumerate}
    \item $u$, $w$, and $p$ are accepted
    \item $u$, $w$, and $v$ are accepted
    \item $u$, $w$, $p$, and $v$ are accepted
    \item $p$ and one of $u$ and $w$ are accepted
    \item $v$ and one of $u$ and $w$ are accepted
    \item $p$, $v$, and one of $u$ and $w$ are accepted
\end{enumerate}

Furthermore, define a half edge $(u, v)$ as being \emph{free} if it is halvable both by accepting $u$ and by accepting $v$.
Also, define a half edge as being \emph{expired} if it is neither $\me1$ nor $\me2$.
Finally, define $A$ as the set of vertices accepted by $\PEA$.
Then, define the potential function $\Phi$ as
\begin{equation*}
\Phi := |\{(u, v) | (u, v) \mbox{ expired half}\}| + \frac{1}{3}|A \cap (\me1 \cup \me2)| + \frac{2}{3} \cdot \mathds{1}[\mbox{free half } \me2]
\end{equation*}

Now, we show that, for any possible state transition triggered by a newly-revealed vertex, the number of incurred late operations $\lo$ added to the change in potential $\dphi$ is bounded above by $\frac{10}{3}$.
Note that, for any newly-revealed vertex $v$, $v$ may be adjacent to $k \ge 0$ rejected vertices that are matched by some expired edge.
This incurs $k$ late operations, but also decreases $\Phi$ by $k$, so this may be ignored when computing $\lo + \dphi$.

\begin{itemize}
    \item $1 \to 1$, without shift: $\lo + \dphi \le 0 + (0 + 0 + 0) = 0$
    \item $1 \to 1$, with shift: not possible, because an edge shifting from $\me1$ to $\me2$ won't go from half to full (Lemma~\ref{lm:shift-me2}) 
    \item $1 \to 2$, without shift: $\lo + \dphi \le 2 + (0 + 0 + 0) = 2$
    \item $1 \to 2$, with shift: not possible, because an edge shifting from $\me1$ to $\me2$ won't go from half to full (Lemma~\ref{lm:shift-me2}) 
    \item $1 \to 3$, without shift: $\lo + \dphi \le 1 + (0 + \frac{1}{3} + 0) = \frac{4}{3}$
    \item $1 \to 3$, with shift: not possible, because $\me1$ is always halvable after a shift (Observation~\ref{obs:shift-me1}) 
    \item $1 \to 4$, without shift: $\lo + \dphi \le 1 + (0 - \frac{1}{3} + \frac{2}{3}) = \frac{4}{3}$
    \item $1 \to 4$, with shift: $\lo + \dphi \le 1 + (0 - \frac{1}{3} + \frac{2}{3}) = \frac{4}{3}$
    \item $1 \to 5$, without shift: $\lo + \dphi \le 3 + (0 - \frac{1}{3} + \frac{2}{3}) = \frac{10}{3}$
    \item $1 \to 5$, with shift: $\lo + \dphi \le 0 + (0 - \frac{1}{3} + \frac{2}{3}) = \frac{1}{3}$
    \item $1 \to 6$, without shift: not possible, because $\PEA$ won't halve $\me2$ if $\me1$ is full 
    \item $1 \to 6$, with shift: not possible, because $\PEA$ prioritizes halving $\me1$ and this is always possible after a shift (Observation~\ref{obs:shift-me1}) 
    
    \item $2 \to 1$, without shift: $\lo + \dphi \le 2 + (0 + 0 + 0) = 2$
    \item $2 \to 1$, with shift: not possible, because an edge shifting from $\me1$ to $\me2$ won't go from half to full (Lemma~\ref{lm:shift-me2}) 
    \item $2 \to 2$, without shift: $\lo + \dphi \le 0 + (0 + 0 + 0) = 0$
    \item $2 \to 2$, with shift: not possible, because an edge shifting from $\me1$ to $\me2$ won't go from half to full (Lemma~\ref{lm:shift-me2}) 
    \item $2 \to 3$, without shift: $\lo + \dphi \le 1 + (0 + \frac{1}{3} + 0) = \frac{4}{3}$
    \item $2 \to 3$, with shift: not possible, because $\me1$ is always halvable after a shift (Observation~\ref{obs:shift-me1}) 
    \item $2 \to 4$, without shift: $\lo + \dphi \le 3 + (0 - \frac{1}{3} + \frac{2}{3}) = \frac{10}{3}$
    \item $2 \to 4$, with shift: $\lo + \dphi \le 1 + (0 - \frac{1}{3} + \frac{2}{3}) = \frac{4}{3}$
    \item $2 \to 5$, without shift: $\lo + \dphi \le 1 + (0 - \frac{1}{3} + \frac{2}{3}) = \frac{4}{3}$
    \item $2 \to 5$, with shift: $\lo + \dphi \le 0 + (0 - \frac{1}{3} + \frac{2}{3}) = \frac{1}{3}$
    \item $2 \to 6$, without shift: not possible, because $\PEA$ won't halve $\me2$ if $\me1$ is full 
    \item $2 \to 6$, with shift: not possible, because $\PEA$ prioritizes halving $\me1$ and this is always possible after a shift (Observation~\ref{obs:shift-me1}) 
    
    \item $3 \to 1$, without shift: $\lo + \dphi \le 1 + (0 - \frac{1}{3} + 0) = \frac{2}{3}$
    \item $3 \to 1$, with shift: $\lo + \dphi \le 1 + (0 - \frac{1}{3} + 0) = \frac{2}{3}$
    \item $3 \to 2$, without shift: $\lo + \dphi \le 1 + (0 - \frac{1}{3} + 0) = \frac{2}{3}$
    \item $3 \to 2$, with shift: $\lo + \dphi \le 0 + (0 - \frac{1}{3} + 0) = -\frac{1}{3}$
    \item $3 \to 3$, without shift: $\lo + \dphi \le 0 + (0 + 0 + 0) = 0$
    \item $3 \to 3$, with shift: not possible, because $\me1$ is always halvable after a shift (Observation~\ref{obs:shift-me1}) 
    \item $3 \to 4$, without shift: $\lo + \dphi \le 2 + (0 - \frac{2}{3} + \frac{2}{3}) = 2$
    \item $3 \to 4$, with shift: $\lo + \dphi \le 2 + (0 - \frac{2}{3} + \frac{2}{3}) = 2$
    \item $3 \to 5$, without shift: $\lo + \dphi \le 2 + (0 - \frac{2}{3} + \frac{2}{3}) = 2$
    \item $3 \to 5$, with shift: $\lo + \dphi \le 1 + (0 - \frac{2}{3} + \frac{2}{3}) = 1$
    \item $3 \to 6$, without shift: not possible, because $\PEA$ won't halve $\me2$ if $\me1$ is full 
    \item $3 \to 6$, with shift: not possible, because $\PEA$ prioritizes halving $\me1$ and this is always possible after a shift (Observation~\ref{obs:shift-me1}) 
    
    \item $4 \to 1$, without shift: $\lo + \dphi \le 1 + (0 + \frac{1}{3} + 0) = \frac{4}{3}$
    \item $4 \to 1$, with shift: not possible, because an edge shifting from $\me1$ to $\me2$ won't go from half to full (Lemma~\ref{lm:shift-me2}) 
    \item $4 \to 2$, without shift: $\lo + \dphi \le 3 + (0 + \frac{1}{3} + 0) = \frac{10}{3}$
    \item $4 \to 2$, with shift: not possible, because an edge shifting from $\me1$ to $\me2$ won't go from half to full (Lemma~\ref{lm:shift-me2}) 
    \item $4 \to 3$, without shift: $\lo + \dphi \le 2 + (0 + \frac{2}{3} + 0) = \frac{8}{3}$
    \item $4 \to 3$, with shift: not possible, because $\me1$ is always halvable after a shift (Observation~\ref{obs:shift-me1}) 
    \item $4 \to 4$, without shift: $\lo + \dphi \le 2 + (0 + 0 - \frac{2}{3}) = \frac{4}{3}$ 
    \item $4 \to 4$, with shift: $\lo + \dphi \le 1 + (1 + 0 + \frac{2}{3}) = \frac{8}{3}$
    \item $4 \to 5$, without shift: $\lo + \dphi \le 4 + (0 + 0 - \frac{2}{3}) = \frac{10}{3}$ 
    \item $4 \to 5$, with shift: $\lo + \dphi \le 0 + (1 + 0 + \frac{2}{3}) = \frac{5}{3}$
    \item $4 \to 6$, without shift: $\lo + \dphi \le 3 + (0 + \frac{1}{3} - \frac{2}{3}) = \frac{8}{3}$ 
    \item $4 \to 6$, with shift: not possible, because $\PEA$ prioritizes halving $\me1$ and this is always possible after a shift (Observation~\ref{obs:shift-me1}) 
    
    \item $5 \to 1$, without shift: $\lo + \dphi \le 3 + (0 + \frac{1}{3} + 0) = \frac{10}{3}$
    \item $5 \to 1$, with shift: not possible, because an edge shifting from $\me1$ to $\me2$ won't go from half to full (Lemma~\ref{lm:shift-me2}) 
    \item $5 \to 2$, without shift: $\lo + \dphi \le 1 + (0 + \frac{1}{3} + 0) = \frac{4}{3}$
    \item $5 \to 2$, with shift: not possible, because an edge shifting from $\me1$ to $\me2$ won't go from half to full (Lemma~\ref{lm:shift-me2}) 
    \item $5 \to 3$, without shift: $\lo + \dphi \le 2 + (0 + \frac{2}{3} + 0) = \frac{8}{3}$
    \item $5 \to 3$, with shift: not possible, because $\me1$ is always halvable after a shift (Observation~\ref{obs:shift-me1}) 
    \item $5 \to 4$, without shift: $\lo + \dphi \le 4 + (0 + 0 - \frac{2}{3}) = \frac{10}{3}$ 
    \item $5 \to 4$, with shift: $\lo + \dphi \le 1 + (1 + 0 + \frac{2}{3}) = \frac{8}{3}$
    \item $5 \to 5$, without shift: $\lo + \dphi \le 2 + (0 + 0 - \frac{2}{3}) = \frac{4}{3}$ 
    \item $5 \to 5$, with shift: $\lo + \dphi \le 0 + (1 + 0 + \frac{2}{3}) = \frac{5}{3}$
    \item $5 \to 6$, without shift: $\lo + \dphi \le 3 + (0 + \frac{1}{3} - \frac{2}{3}) = \frac{8}{3}$ 
    \item $5 \to 6$, with shift: not possible, because $\PEA$ prioritizes halving $\me1$ and this is always possible after a shift (Observation~\ref{obs:shift-me1}) 
    
    \item $6 \to 1$, without shift: $\lo + \dphi \le 2 + (0 + 0 + 0) = 2$
    \item $6 \to 1$, with shift: $\lo + \dphi \le 1 + (1 + 0 + 0) = 2$
    \item $6 \to 2$, without shift: $\lo + \dphi \le 2 + (0 + 0 + 0) = 2$
    \item $6 \to 2$, with shift: $\lo + \dphi \le 0 + (1 + 0 + 0) = 1$
    \item $6 \to 3$, without shift: $\lo + \dphi \le 1 + (0 + \frac{1}{3} + 0) = \frac{4}{3}$
    \item $6 \to 3$, with shift: not possible, because $\me1$ is always halvable after a shift (Observation~\ref{obs:shift-me1}) 
    \item $6 \to 4$, without shift: $\lo + \dphi \le 3 + (0 - \frac{1}{3} - \frac{2}{3}) = 2$ 
    \item $6 \to 4$, with shift: $\lo + \dphi \le 2 + (1 - \frac{1}{3} + \frac{2}{3}) = \frac{10}{3}$
    \item $6 \to 5$, without shift: $\lo + \dphi \le 3 + (0 - \frac{1}{3} - \frac{2}{3}) = 2$ 
    \item $6 \to 5$, with shift: $\lo + \dphi \le 1 + (1 - \frac{1}{3} + \frac{2}{3}) = \frac{7}{3}$
    \item $6 \to 6$, without shift: $\lo + \dphi \le 2 + (0 + 0 - \frac{2}{3}) = \frac{4}{3}$ 
    \item $6 \to 6$, with shift: not possible, because $\PEA$ prioritizes halving $\me1$ and this is always possible after a shift (Observation~\ref{obs:shift-me1}) 
\end{itemize}

Thus, for any possible state transition, $\lo + \dphi \le \frac{10}{3}$.
Furthermore, $\Phi_0 = 0$ and $\Phi_i \ge 0$.
Therefore, the amortized recourse incurred by $\PEA$ is bounded above by $\frac{10}{3}$.
\end{proof}
}
}
\end{document}